\documentclass[12pt,a4paper]{article}
\usepackage{amssymb,amsfonts,amsthm,amsmath}
\usepackage{textcomp}
\usepackage{setspace}
\usepackage{graphicx}
\usepackage{epstopdf}
\usepackage{array}
\usepackage{color}

\usepackage{graphicx}

\def\hang{\hangindent\parindent}
 \def\rf{\par\noindent\hang}

\newtheorem{theorem}{Theorem}

\newtheorem{lemma}{Lemma}

\newcommand{\bX}{\boldsymbol{X}}
\newcommand{\bx}{\boldsymbol{x}}
\newcommand{\bY}{\boldsymbol{Y}}

\newcommand{\bZ}{\boldsymbol{Z}}
\newcommand{\bU}{\boldsymbol{U}}
\newcommand{\btheta}{\boldsymbol{\theta}}

\newcommand{\ntheta}{||\btheta||}
\newcommand{\aT}{a(T)}
\newcommand{\at}{a(t)}
\newcommand{\bT}{b(T)}
\newcommand{\bt}{b(t)}
\newcommand{\aplsT}{a^{+}(T)}
\newcommand{\aplst}{a^{+}(t)}

\newcommand{\subsup}[1]{\mbox{\footnotesize $#1$}}
\newcommand{\pkstyle}[1]{\mbox{\footnotesize $#1$}}

\begin{document}

\baselineskip=17pt


\noindent {\LARGE {\bf Optimized recentered confidence spheres for the multivariate normal mean}}

\vspace{1cm}



\noindent {\Large{Waruni Abeysekera, Paul Kabaila$^*$}}

\bigskip

\noindent
{\sl Department of Mathematics and Statistics, La Trobe University, Victoria 3086, Australia}

\vspace{0.5cm}

\noindent  \textbf{ABSTRACT}

\medskip

\noindent Casella and Hwang, 1983, {\sl JASA}, introduced a broad class of recentered confidence
spheres for the mean $\btheta$
of a multivariate normal distribution with covariance matrix $\sigma^2 \boldsymbol{I}$, for $\sigma^2$
known. Both the center and radius functions of these confidence spheres are flexible functions of
the data.
For the particular case of confidence spheres centered on the positive-part James-Stein estimator and with radius determined by empirical Bayes
considerations, they show numerically that these confidence spheres have the desired minimum coverage probability $1-\alpha$ and dominate the
usual confidence sphere in terms of scaled volume. We shift the focus from the scaled volume to the scaled {\sl expected} volume of the recentered confidence
sphere. Since both the coverage probability and the scaled expected volume are functions of the Euclidean norm of $\btheta$,
it is feasible to optimize the performance of the recentered confidence sphere by numerically computing both the
center and radius functions so as to optimize some clearly specified criterion.
We suppose that we have uncertain prior information that $\btheta = \boldsymbol{0}$. This motivates us to determine
the center and radius functions of the confidence sphere by numerical minimization of the scaled expected volume of the confidence sphere at $\btheta = \boldsymbol{0}$, subject to the constraints that (a) the coverage probability never falls below $1-\alpha$ and (b) the radius never exceeds
the radius of the standard $1-\alpha$ confidence sphere. Our results show that, by focusing on this clearly specified criterion,
significant gains in performance (in terms of this criterion) can be achieved.
We also present analogous results for the
much more difficult case that $\sigma^2$ is unknown.

\bigskip


\vspace{0.5cm}


\noindent  {\large $^*$} Correspondence to:
Department of Mathematics and Statistics, La Trobe University, Victoria 3086, Australia.

\smallskip

\noindent {\it E-mail addresses:} abeysekera.w@wehi.edu.au (W. Abeysekera), P.Kabaila@latrobe.edu.au (P. Kabaila)

\newpage

\noindent {\large{\bf 1. Introduction}}

\medskip

\noindent Suppose that $\bX = \left( X_1,...,X_p \right) \sim N\left( \btheta, \sigma^2 \boldsymbol{I} \right)$
where
$\btheta = \left( \theta_1,...,\theta_p \right)$ and
$\boldsymbol{I}$ denotes the $p \times p$ identity matrix ($p \ge 3$).
Stein (1962) presents arguments that suggest that, for $\sigma^2$ known, a confidence sphere centered on the
positive-part James-Stein estimator and with the same radius as the standard $1-\alpha$ confidence sphere
for $\btheta$ dominates the standard $1-\alpha$ confidence sphere in terms of coverage probability.
This remarkable suggested result (proved later by Hwang and Casella, 1982)
mirrors the earlier results on the point estimation of $\btheta$ known as Stein's paradox.

Casella and Hwang (1983, Section 3)
introduce a broad class of recentered confidence
spheres for $\btheta$, for $\sigma^2$
known. Both the center and radius functions of these confidence spheres are flexible functions of
the data.
For the particular case of confidence spheres centered on the positive-part James-Stein estimator and with radius determined by empirical Bayes
considerations, they show numerically that, for sufficiently large $p$, these confidence spheres have the desired minimum coverage probability $1-\alpha$ and dominate the
usual confidence sphere in terms of the scaled volume.
For $\sigma^2$ known, Samworth (2005) also considered a recentered confidence sphere (RCS)
with center at the positive-part James-Stein estimator.
However, he determines the radius function using either a Taylor series or the bootstrap.
He shows numerically that these confidence spheres have the desired minimum coverage probability $1-\alpha$ for sufficiently large $p$ and dominate the usual confidence sphere in terms of the $p$'th root of the scaled volume.
In common with much of the existing literature,
we first consider the case that $\sigma^2$ is known. Later, we consider the more difficult case that
$\sigma^2$ is unknown.

\textbf{Suppose that $\boldsymbol{\sigma^2}$ is known}.
We shift the focus from the scaled volume (or its $p$'th root) to the scaled {\sl expected} volume of the RCS.
Scaled expected length has been profitably used in related problems and to resolve a paradox in decision-theoretic
interval estimation (Farchione and Kabaila, 2008, Kabaila and Giri, 2009,
Kabaila and Tissera, 2014, and Kabaila, 2013).
Casella, Hwang and Robert (1993) show that a confidence interval for the univariate normal mean that is obtained by minimizing the posterior expected loss, for the prior distribution and the risk function that they specify, has paradoxical properties. 
Kabaila (2013) shows that these paradoxical properties disappear when the expected length term in this risk function is replaced by the scaled expected length.
Since both the coverage probability and the scaled expected volume of the RCS are functions of $\| \btheta \|$,
it is feasible to optimize the performance of the RCS by numerically computing both the
center and radius functions so as to optimize some clearly specified criterion, subject to coverage and radius constraints.
By contrast, a goal of seeking to minimize
(in some sense) the scaled volume of the recentered confidence sphere for the most probable values of
$\bX$ when $\btheta = \boldsymbol{0}$, subject to the coverage constraint, is problematic (Casella and Hwang, 1986).

Casella and Hwang (1987) argue cogently that the confidence set for $\btheta$ should be
tailored to the uncertain prior information available about $\btheta$.
We suppose that we have uncertain prior information that $\btheta = \boldsymbol{0}$. 
Hodges and Lehmann 
(1952) propose, quite broadly, the utilization of uncertain prior information in frequentist inference. Our utilization of the uncertain 
prior information that $\btheta = \boldsymbol{0}$ is also frequentist.
This uncertain prior information motivates us to determine
the center and radius functions of the RCS by numerical minimization of the scaled expected volume of the confidence sphere at $\btheta = \boldsymbol{0}$, subject to the constraints that (a) the coverage probability never falls below $1-\alpha$ and (b) the radius never exceeds
the radius of the standard $1-\alpha$ confidence sphere (centered on $\bX$).
The numerical results in Section 2
show that, by focusing on the clearly specified criterion of the scaled expected  volume of the confidence sphere at $\btheta = \boldsymbol{0}$,
significant gains in performance (in terms of this specified criterion)
can be achieved.

Of course, our approach requires the use of a computationally convenient formula for the coverage probability of the RCS.
Such a formula is derived by Casella and Hwang (1983, Section 3), for $p$ odd. To be able to compute the coverage probability also for
$p$ even, we derive a new computationally convenient formula for the coverage probability of the RCS that is applicable for both even and odd $p$.
The coverage constraint is implemented in the computations by requiring that this constraint is satisfied for a judiciously chosen finite set of
values of $\| \btheta \|$. To show that a given finite set is adequate to the task, we simply check that at the completion of the computations
of the optimized RCS, the coverage probability constraint is satisfied for all $\| \btheta \|$.
For computational feasibility, we also need to choose parametric forms for the center and radius functions.
This choice is by no means obvious and, as described in Section 2
(see, particularly, Remark 2.1), requires a great deal of care.

A natural requirement for any confidence set for $\btheta$ is that this it is 
{\sl rotationally symmetric}. The optimized RCS's that we compute satisfy this requirement. Efron (2006) provides an elegant description of any rotationally symmetric confidence set in terms of his `inclusion function'. 
This is a function of only two variables: $\| \btheta \|$ and $\| \bx \|$.
In Section 3, we compare the graphs of the inclusion functions for (a) the standard confidence sphere, (b) the RCS of Casella and Hwang (1983) and (c) the optimized RCS.

\textbf{Now consider the more difficult case that $\boldsymbol{\sigma^2}$ is unknown}.
Suppose that we have additional data that provides the estimator $S^2$ for $\sigma^2$,
where $m S^2/\sigma^2 \sim \chi^2_m$ and $S^2$ and $\bX$ are independent.
In the related context that there is uncertain prior information that $\theta_1 = \theta_2 = \dots = \theta_p$,
Casella and Hwang (1987) put forward an RCS
with center at an analogue of the positive-part James-Stein estimator (which is defined for $\sigma^2$ known)
and radius that is an analogue of the radius based on empirical Bayes considerations for $\sigma^2$ known.

In Section 4, we describe a class of RCS's that are an analogue, for $\sigma^2$ unknown, of the broad class of
RCS's described by Casella and Hwang (1983), Section 3, for $\sigma^2$ known.
Both the coverage probability and the scaled expected volume of the RCS's in this class are functions of $\gamma = \| \btheta \|/\sigma$.
As before, suppose that we have uncertain prior information that $\btheta = \boldsymbol{0}$.
Again, this motivates us to determine
the center and radius functions of the RCS by numerical minimization of the scaled expected volume of the confidence sphere at $\btheta = \boldsymbol{0}$, subject to the constraints that (a) the coverage probability never falls below $1-\alpha$ and (b) the radius never exceeds
the radius of the standard $1-\alpha$ confidence sphere (centered on $\bX$).
The numerical results in Section 4 show that, by focusing on the clearly specified criterion of the scaled expected  volume of the confidence sphere at $\btheta = \boldsymbol{0}$,
significant gains in performance can be achieved, by comparison with the RCS centered on the analogue of the positive-part James-Stein
estimator.

\bigskip

\noindent {\large{\bf 2. Results for $\boldsymbol{\sigma^2}$ known.
Comparison of the
performances of the optimized RCS and the RCS of Casella and Hwang (1983, Section 4).}}

\medskip
In this section, we suppose that $\sigma^2$ is known.
Without loss of generality, we assume that $\sigma^2 = 1$.
The standard $1-\alpha$ confidence set for $\btheta$ is
$I = \{\btheta : \, \| \btheta-\bX \| \leq d \}$,
where the positive number $d$ satisfies
$P\big( Q \leq d^2 \big) = 1-\alpha$
for $Q \sim \chi_p^2$.
Casella and Hwang (1983, Section 3), define a class of RCS's that can be
expressed in the form
\begin{equation*}
J(a,b) = \big\{ \btheta : \| \aT \bX - \btheta \| \leq \bT \big\},
\end{equation*}
where $a:[0, \infty) \rightarrow (0, \infty)$, $b:[0, \infty) \rightarrow (0, \infty)$ and
$T = \| \bX \|/\sqrt{p}$.
This notation for the RCS is slightly different
from that used by Casella and Hwang (1983),
who express this RCS in
terms of $\| \bX \|$. This makes no essential difference.
This choice of center and radius has some intuitive appeal, since
$T = \|\bX\|/\sqrt{p}$ may be viewed as a test statistic for testing the null hypothesis
that $\btheta = \boldsymbol{0}$ against the
alternative hypothesis that $\btheta \ne \boldsymbol{0}$.
We assess the RCS $J(a,b)$ using both its coverage probability an its
scaled expected volume, which is defined
to be the ratio (expected volume of the RCS) / (volume of $I$).

Casella and Hwang (1983, Section 3), derive a computationally convenient formula for the coverage
probability of $J(a,b)$ that is applicable for $p$ odd.
Let $\gamma = \|\btheta \|$.
In Appendix
A, we show that
the coverage probability of $J(a,b)$ is, for given functions $a$ and $b$, a function of
$\gamma$ and we
derive a new computationally convenient formula for this coverage probability that is applicable for any $p$ (even or odd).
Details of the numerical evaluation of this coverage probability, using these computationally convenient formulas,
are also presented in Appendix A. The numerical results for coverage probabilities that are presented in this section
were found using this new computationally convenient formula.

Define $a^+:[0, \infty) \rightarrow (0, \infty)$ by the requirement that
$a^+(T) \bX$  is the positive-part James-Stein estimator. This implies that
\begin{equation*}
a^{+}(x)=\max \left\{ 0, 1-\left(1-\frac{2}{p}\right)\frac{1}{x^2} \right\}.
\end{equation*}
The specific proposal for an RCS that is given in Section 4 of Casella and Hwang (1983)
is $J(a^+,b^*)$, where $b^*$ is determined by empirical Bayes considerations.
For $x \in [0, d/\sqrt{p}]$, 
\begin{equation*}
b^*(x) =
\sqrt{\{ 1-(p-2)/d^2 \}[ d^2 - p \, \log \{ 1-(p-2)/d^2 \}]}  
\end{equation*}
and, for $x > d/\sqrt{p}$,
\begin{equation*}
b^*(x) =
\sqrt{\{ 1-(p-2)/(p \, x^2) \} [ d^2 - p \, \log \{ 1-(p-2)/(p \, x^2) \} ]}. 
\end{equation*}

We define the scaled expected volume of $J(a,b)$ to be the ratio
\begin{equation}
\label{sev_initial}
\frac{E_{\btheta} \{\text{volume of } J(a,b)\}}{\text{volume of } I}
= E_{\btheta} \left\{ \left(\frac{b(T)}{d} \right)^p \right\},
\end{equation}
since the volume of a sphere in $\mathbb{R}^p$ with radius $r$ is $2 \, r^p \, \pi^{p/2}/\big \{p \, \Gamma(p/2)\big \}$.
In Appendix A, we show that this is a function of $\gamma = \|\btheta \|$, for given function $b$. We also
derive a new computationally convenient formula for this scaled expected volume.
To find the optimized RCS, we require
that the functions $a$ and $b$ satisfy the following conditions.

\smallskip

\noindent \underbar{Condition A} \ $a:[0, \infty) \rightarrow (0, \infty)$ is a continuous
nondecreasing function
that satisfies $a(x)=a^{+}(x)$ for all $x \geq k$, where $a^+(T) \bX$  is the positive-part James-Stein estimator
and $k$ is a (sufficiently large) specified positive number.

\smallskip

\noindent \underbar{Condition B} \ $b:[0, \infty) \rightarrow (0, \infty)$ is a continuous
nondecreasing function
that satisfies $b(x)=d$ for all $x \geq k$.

\smallskip

In addition, for computational feasibility, we specify the following parametric forms for these
functions.
\begin{enumerate}
	
	\item
	
	Suppose that $x_1,\dots,x_{q_1}$ satisfy $0 = x_1 < x_2 < \dots < x_{q_1} = k$. The function $a$ is fully
	specified by the vector $a(x_1),\dots,a(x_{q_1})$ as follows. The value of $a(x)$ for any given
	$x \in [0, k]$ is found by piecewise cubic Hermite polynomial interpolation
	for these given function values. We call $x_1, \dots,x_{q_1}$ the knots
	of this piecewise cubic Hermite polynomial.
	
	\item
	
	Suppose that $y_1,\dots,y_{q_2}$ satisfy $0 = y_1 < y_2 < \dots < y_{q_2} = k$. The function $b$ is fully
	specified by the vector $b(y_1),\dots,b(y_{q_2})$ as follows. The value of $b(y)$ for any given
	$y \in [0, k]$ is found by piecewise cubic Hermite polynomial interpolation
	for these given function values. We call $y_1, \dots,y_{q_2}$ the knots
	of this piecewise cubic Hermite polynomial.
	
\end{enumerate}

For judiciously-chosen values of $k$ and these knots, we compute the functions $a$ and $b$, which take
these parametric forms, are nondecreasing and are
such that (a) the scaled expected
volume evaluated at $\btheta = \boldsymbol{0}$ (i.e. at $\gamma= 0$) is minimized
and (b) the coverage
probability of $J(a,b)$ never falls below $1-\alpha$.
All of the computations presented in the present paper were performed using
programs written in MATLAB using the Statistics and
Optimization toolboxes.
Piecewise cubic Hermite interpolation (Fritsch and Carlson, 1980)
is implemented in the pchip function in MATLAB.

The coverage constraint is implemented in the
computations as follows. For any reasonable choice of the functions $a$ and $b$,
the coverage probability of $J(a,b)$ converges to $1-\alpha$ as $\gamma \rightarrow \infty$. The
constraints implemented in the computations are that the coverage probability of $J(a,b)$ is
greater than or equal to $1-\alpha$ for every $\gamma$ in a judiciously-chosen finite set of
values. That a given finite set of values of $\gamma$ is adequate to the task is judged by checking
numerically,
at the completion of computations, that the coverage probability constraint is satisfied
for all $\gamma \ge 0$.

For $1-\alpha = 0.95$, we compare the coverage probability and scaled expected volume
of the optimized RCS with $J(a^+, b^*)$, the RCS of Casella and Hwang (1983, Section 4).
We chose the knots of $a$ and $b$ that allow these functions
to provide good approximations to $a^+$ and $b^*$, respectively. In this way,
we sought to ensure that $J(a,b)$ could perform at least as well as $J(a^+, b^*)$
in terms of minimizing the scaled expected volume at $\gamma = 0$, subject to the
coverage and radius constraints.
Some exploratory computations led us to choose $k=10$ and the following knots for
$a$ and $b$.
Since $a^{+}(x)=0$ for $0 \leq x \leq \sqrt{1-(2/p)}$,
we place the first two knots of the function $a$ at $0$ and $\sqrt{1-(2/p)}$.
The next three knots of $a$ are at $\sqrt{1-(2/p)}+(\tau/10)$, $\sqrt{1-(2/p)}+(2\tau/10)$ and $\sqrt{1-(2/p)}+(4\tau/10)$,
where $\tau= (k/2)-\sqrt{1-(2/p)}$. The remaining knots of $a$ are at $k/2$, $3k/4$ and $k$.
Since $b^{+}(x)$ is a constant for $0 \leq x \leq d/\sqrt{p}$,
we place the first two knots of the function $b$ at $0$ and $d/\sqrt{p}$.
The next two knots of $b$ are at $d/\sqrt{p}+(\xi/3)$ and $d/\sqrt{p}+(2\xi/3)$,
where $\xi= (k/2)-d/\sqrt{p}$. The remaining knots of $b$ are at $k/2$, $3k/4$ and $k$.
The optimized RCS was computed for each $p \in \{3,4,\dots,13,20,25\}$.

The coverage constraint was implemented in the computations by requiring that the
coverage probability of $J(a,b)$ is greater than or equal to $1-\alpha$ for all
$\gamma \in \{0,1,2,\dots,64,65\}$. This was shown to be adequate to the task by checking
numerically, at the completion of the computation of the optimized RCS,
that the coverage probability constraint
is satisfied for all $\gamma \ge 0$.

Figure \ref{Figure1} shows that, for $p=3$, the coverage probability of the optimized RCS is no less than $0.95$ for all
$\gamma$, while the coverage probability of $J(a^+, b^*)$, the RCS of Casella and Hwang (1983, Section 4),
is slightly below $0.95$ for some values of $\gamma$. This figure also shows that, for $p=3$, the
scaled expected volume of the optimized RCS is substantially
less than the scaled expected volume of $J(a^+, b^*)$, at $\btheta = \boldsymbol{0}$.
The top two panels of this figure suggest the following from the point of view of minimizing the 
scaled expected volume at $\btheta = \boldsymbol{0}$, subject to the coverage and other constraints.
The shrinkage towards the origin of the center of the RCS of Casella and Hwang (the positive-part James-Stein estimator) is too severe
for small $x$, requiring that the radius of this RCS must be unhelpfully large.

Of course, our optimized RCS does not dominate this RCS of Casella and Hwang. Our optimized RCS has smaller
scaled expected volume for $\gamma$ close to 0. However, it has larger scaled expected volume for $\gamma$
not close to 0.
Table \ref{Table1} presents the comparison of
the minimum coverage probability and the scaled expected volume at $\btheta = \boldsymbol{0}$
of the optimized RCS and $J(a^+, b^*)$ for $p \in \{3,4,\dots,13,20,25\}$.
According to this table, the optimized RCS always achieves a coverage probability greater than or equal to $0.95$,
while $J(a^+, b^*)$, the RCS of Casella and Hwang (1983, Section 4), does not achieve this for $p \leq 6$.
 Also, for every value of $p$ considered, $J(a,b)$ achieves a substantially lower scaled expected volume at
$\btheta = \boldsymbol{0}$ than $J(a^{+},b^*)$. In summary, our optimized RCS compares favourably with that of
Casella and Hwang (1983, Section 4), in terms of both the minimum coverage probability and the
scaled expected volume at $\btheta = \boldsymbol{0}$.

\smallskip
{\sl Remark 2.1:} When we initially considered the construction of optimized RCS's for
$\btheta$, we set $a(x) = 1$ for all $x \ge k$. This seemed a very reasonable choice that leads to
$J(a,b)$ coinciding with the standard $1-\alpha$ confidence set $I$ when
$\|\bX\| \ge k$. Surprisingly, the computation of the nondecreasing functions $a$ and $b$ such that the
scaled expected volume at $\gamma = 0$ was minimized, subject to the coverage probability of $J(a,b)$
never falling below $1-\alpha$, always resulted in a $J(a,b)$ that was, within computational accuracy,
equal to $I$. A careful investigation (Abeysekera, 2014) revealed that the explanation for this phenomenon is that for all
$J(a,b)$'s, other than those very close to $I$, there was a small dip (over a narrow interval of values of
$\gamma$) in the coverage probability below $1-\alpha$. As $k$ is increased, this dip becomes less pronounced,
but appears to never disappear entirely. In other words, it did not seem possible for $J(a,b)$ to satisfy
the coverage constraint unless it was, within computational accuracy, equal to $I$. We found the following
solution to this problem. If, instead of setting $a(x) = 1$ for all $x \ge k$, we set
$a(x) = a^+(x)$ for all $x \ge k$, then this phenomenon does not occur.

%
\begin{figure}[h!]
	\centering
	\begin{tabular}{c}
		\includegraphics[width=0.8\textwidth]{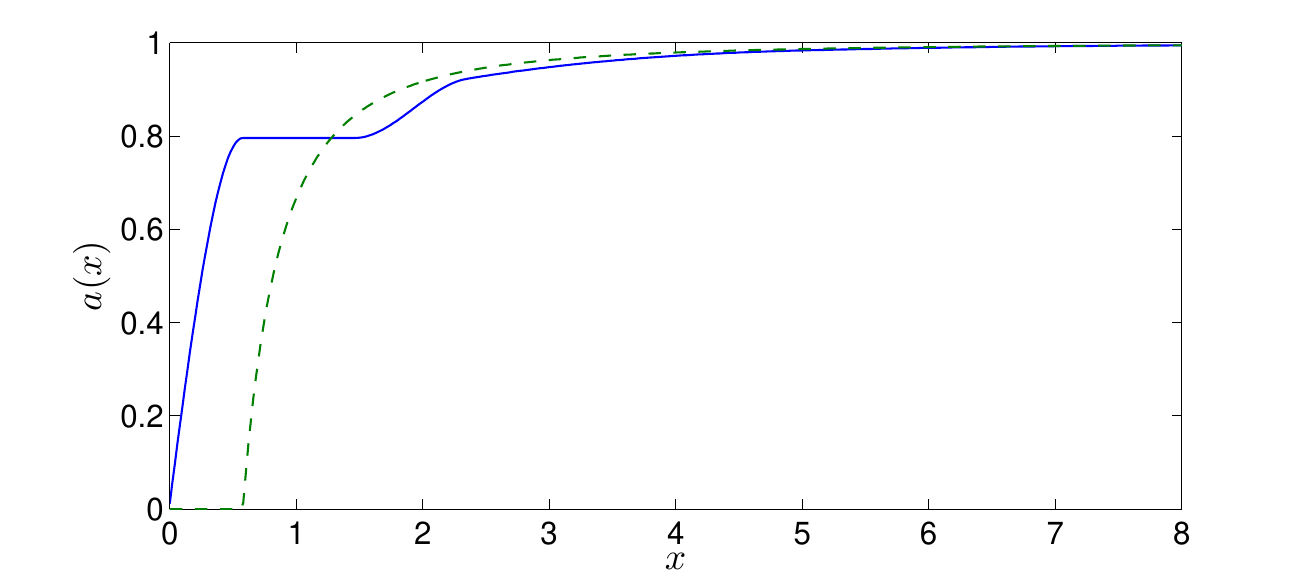} \\
		\includegraphics[width=0.8\textwidth]{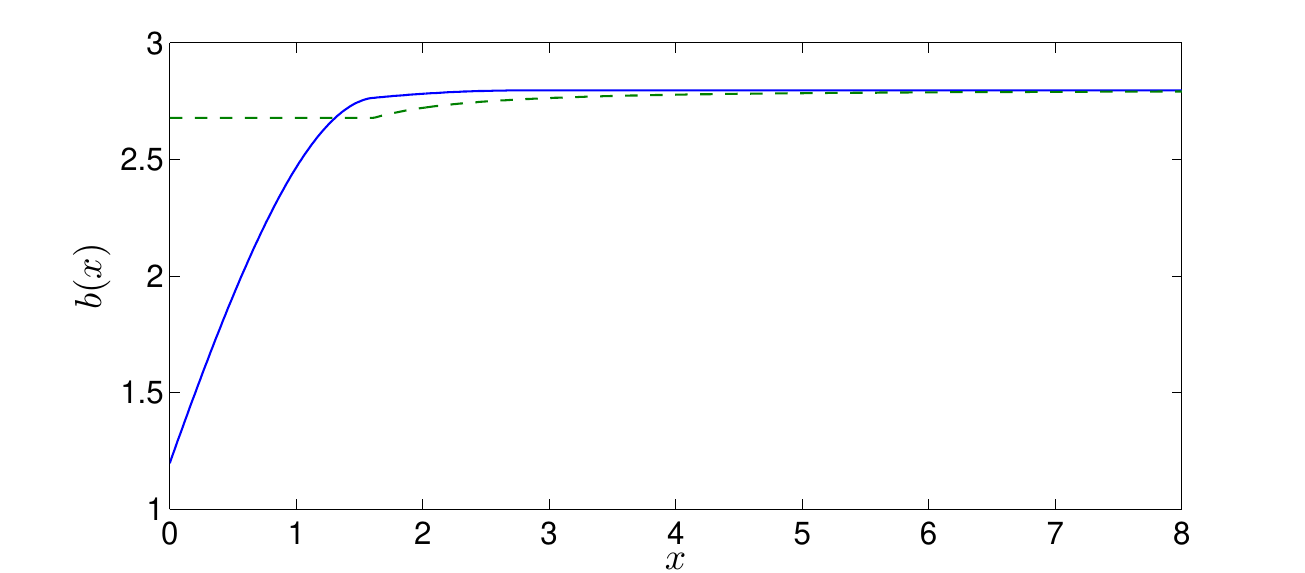}  \\
		\includegraphics[width=0.8\textwidth]{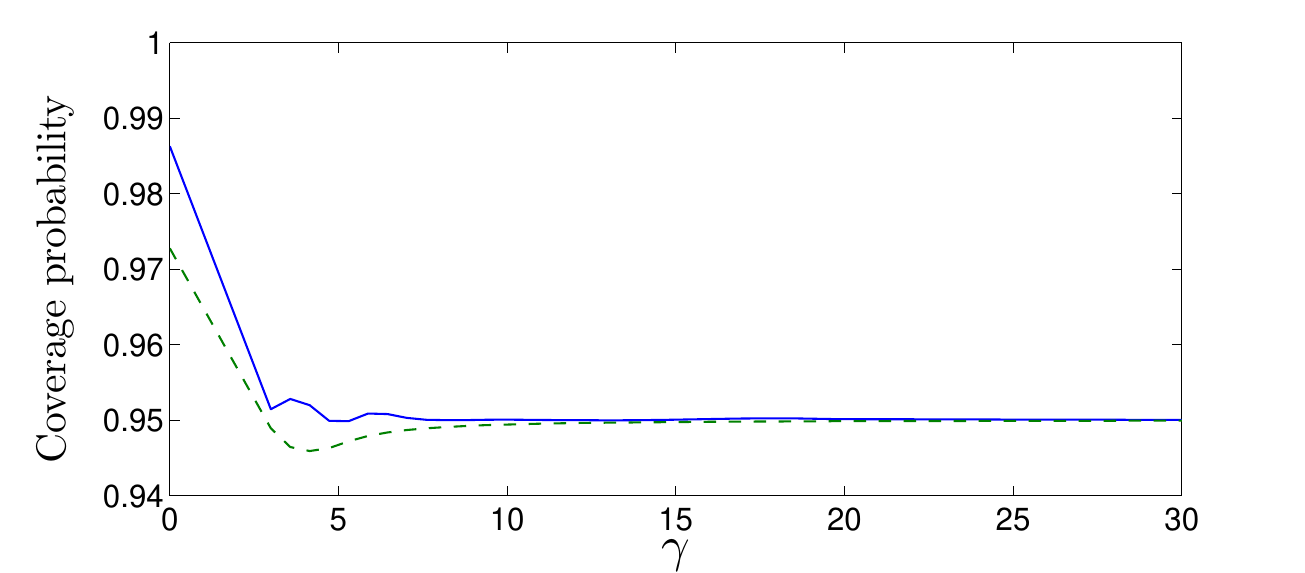} \\
		\includegraphics[width=0.8\textwidth]{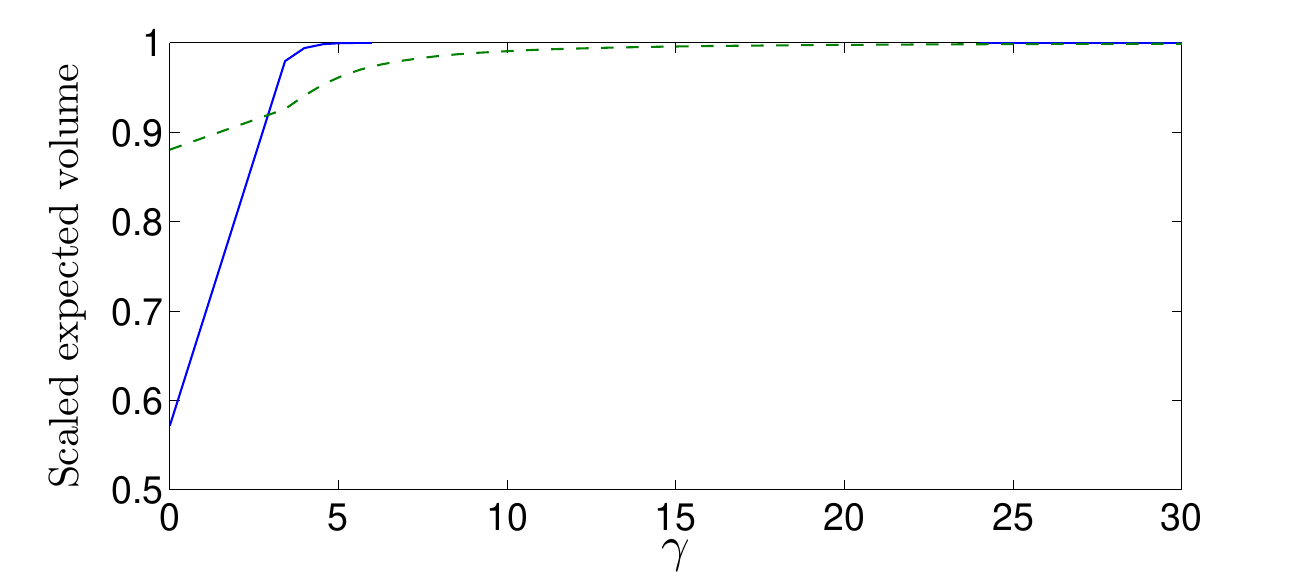} \\
		Legend: ------ optimized RCS\ \ \ - - -  RCS of Casella and Hwang (1983, Section 4).
	\end{tabular}
	\caption{Graphs of the functions $a$ and $b$ and the coverage probability and
		scaled expected volume (as functions of $\gamma = \| \btheta \|$) for both the optimized RCS and
		the RCS of Casella and Hwang (1983, Section 4),
		for $1-\alpha=0.95$ and $p=3$.}
	\label{Figure1}
\end{figure}
\clearpage

\begin{table}[h]
	\label{sig_known_all_results}
	\centering
	\begin{tabular}{>{\centering\arraybackslash}m{1cm} >{\centering\arraybackslash}m{2cm} >{\centering\arraybackslash}m{3cm}
			>{\centering\arraybackslash}m{2cm} >{\centering\arraybackslash}m{2cm}}
		\hline
		$p$ & \multicolumn{2}{c}{RCS of Casella and Hwang}   &  \multicolumn{2}{c}{Optimized RCS}\\
		& \multicolumn{2}{c}{$J(a^{+},b^*)$}          &
		\\
		\hline
		\\
		& minimum    & SEV at                       & minimum    & SEV at  \\
		& CP         & $\btheta = \boldsymbol{0}$   & CP         & $\btheta = \boldsymbol{0}$ \\
		\\
		3	& 0.94594    & 0.88054                      & 0.95       & 0.57155                      \\
		4	& 0.94609    & 0.75553                      & 0.95       & 0.37381                      \\
		5	& 0.94666    & 0.63637                      & 0.95       & 0.23814                      \\
		6	& 0.94852    & 0.52826                      & 0.95       & 0.15362                      \\
		7	& 0.95       & 0.43314                      & 0.95       & 0.09975                      \\
		8	& 0.95       & 0.35142                      & 0.95       & 0.06503                      \\
		9	& 0.95       & 0.28243                      & 0.95       & 0.04221                      \\
		10	& 0.95       & 0.22505                      & 0.95       & 0.02741                      \\
		11	& 0.95       & 0.17794                      & 0.95       & 0.01782                      \\
		12	& 0.95       & 0.13966                      & 0.95       & 0.01155                      \\
		13	& 0.95       & 0.10889                      & 0.95       & 0.00752                      \\
		20	& 0.95       & 0.01629                      & 0.95       & 0.00049                      \\
		25	& 0.95       & 0.00367                      & 0.95       & 0.00004                      \\
		\hline
	\end{tabular}
	\caption{Comparison of the optimized RCS and $J(a^{+},b^*)$,
		the RCS of Casella and Hwang (1983, Section 4),
		with respect to the minimum coverage probability (CP) and the scaled expected volume (SEV)
		at $\btheta = \boldsymbol{0}$,
		for $1-\alpha=0.95$ and $p \in \{3,4,\dots,13,20,25\}$.}
	\label{Table1}
\end{table}

{\sl Remark 2.2:}  We have chosen the functions $a$ and $b$ to be
a piecewise cubic Hermite interpolating polynomial
in the interval $[0,k]$. Other choices of parametric forms for this function are also possible.
For example, one could choose this function to be a quadratic spline in this
interval. Our reason for choosing piecewise cubic Hermite interpolation is that this leads to
interpolating function with fewer undesirable oscillations between the knots than, say, natural
cubic spline interpolation.

\smallskip

{\sl Remark 2.3:} Casella and Hwang (1983, Section 3), argue that it is desirable
that the set $S_{\theta}$, described in their Theorem 3.1, is an interval. During the
computation of the optimized RCS, it was found that at every stage (including the final stage)
this set was an interval.

\bigskip

\noindent {\large{\bf 3. Results for $\boldsymbol{\sigma^2}$ known. Comparison of the inclusion functions the standard confidence sphere, the RCS of Casella and Hwang and the optimized RCS}}

\medskip

Efron (2006) considers confidence sets of the form
\begin{equation*}
\bigcup_{\gamma \ge 0} {\tt SC}_{\bx} \big(\omega_{\gamma}(\| \bx \|), \gamma \big),
\end{equation*}
where $\gamma = \| \btheta \|$ and 
${\tt SC}_{\bx} \big(\omega_{\gamma}(\| \bx \|), \gamma \big)$ is a spherical 
cap of values of $\btheta$ of angular radius $\omega_{\gamma}(\| \bx \|)$
centered at $\gamma \, \bx / \| \bx \|$.
For any rotationally symmetric confidence set we can use this representation to 
find the function $\omega_{\gamma}(\| \bx \|)$. This function can then be used
to find the `inclusion function' $i_{\gamma}(\| \bx \|)$
defined by Efron (2006) to be the conditional probability
\begin{equation*}
P_{\btheta} \Big(\bx \in {\tt SC}_{\btheta} 
\big(\omega_{\gamma}(\| \bx \|), \| \bx \| \big) \, \Big| \, \| \bx \| \Big),
\end{equation*}
where 
${\tt SC}_{\btheta} \big(\omega_{\gamma}(\| \bx \|), \| \bx \| \big)$ is a spherical 
cap of values of $\btheta$ of angular radius $\omega_{\gamma}(\| \bx \|)$
centered at $\| \bx \| \, \btheta  / \gamma$. This conditional density is found
using (2.9) of Efron (2006). The coverage probability of the confidence set is, for any given $\gamma$, $\int_0^{\infty} i_{\gamma}(y) \, f_{\gamma}(y) dy$,
where $f_{\gamma}$ denotes the probability density function of $\| \bX \|$.

In Figures \ref{Figure2} and \ref{Figure3} we compare the inclusion functions of the standard confidence sphere, the RCS of Casella and Hwang (1983, Section 4), and the optimized RCS for $1-\alpha=0.95$. Figures 
\ref{Figure2} and \ref{Figure3}
are for $p=3$ and $p=10$, respectively. 
The top and middle panels of Figure  
\ref{Figure2} 
are for the fairly small values of 
$\gamma = 1.5$ and $\gamma = 2.3$. 
The superiority of the optimized RCS
in terms of scaled expected volume for $\gamma = 0$, is reflected by the fact that,
in these panels, the  inclusion function for the optimized RCS matches up better with the pdf of $\| \bX \|$
than the inclusion functions for both the standard confidence sphere and the RCS of Casella and Hwang. The bottom panel of Figure 2 is for the larger value of $\gamma = 3$. For this larger value, the inclusion functions of both the RCS of Casella and Hwang and the optimized RCS match up equally well (and better than the standard confidence sphere) with the pdf of $\| \bX \|$.

The top and middle panels of Figure 3 are for the fairly small values of 
$\gamma = 2$ and $\gamma = 3.5$. The superiority of the optimized RCS
in terms of scaled expected volume for $\gamma = 0$, is reflected by the fact that,
in these panels, the  inclusion function for the optimized RCS matches up better with the pdf of $\| \bX \|$
than both the inclusion functions for the standard confidence sphere and the RCS of Casella and Hwang. The bottom panel of Figure \ref{Figure3}
is for the larger value of $\gamma = 6$. For this larger value, the inclusion function the RCS of Casella and Hwang matches up 
with the pdf of $\| \bX \|$
somewhat better than the optimized RCS match. Both of these RCS's match up better with
this pdf than the standard confidence sphere. The bottom panel of Figure \ref{Figure3}
compares the inclusion functions for the same values of $p$, $1 - \alpha$ and $\gamma$ as
Figure 2 of Efron (2006).

\begin{figure}[h!]
	\centering
	\includegraphics[width=1\textwidth]{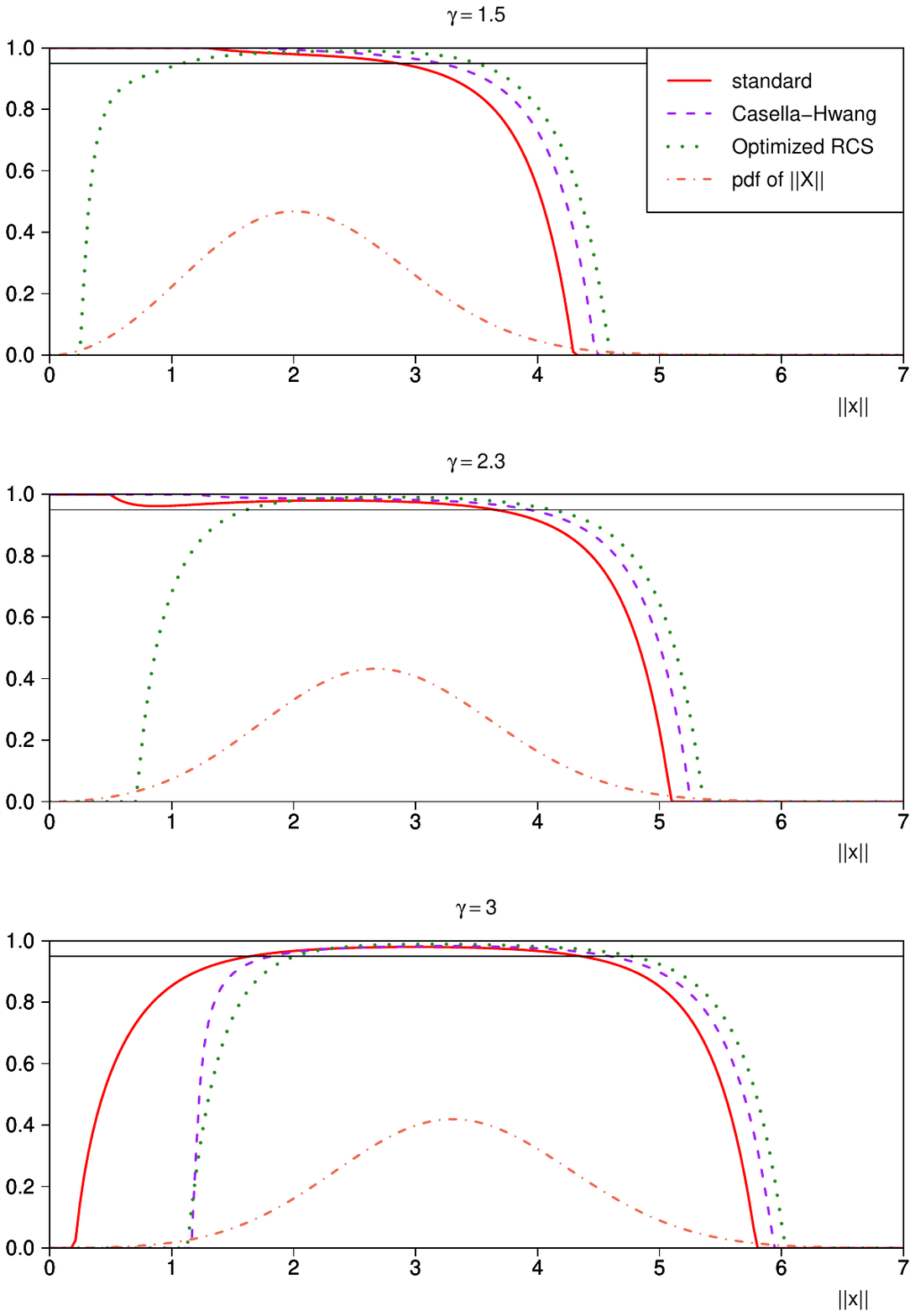} 
	\caption{Each panel consists of graphs of the inclusion functions of the standard confidence sphere, the RCS of Casella and Hwang (1983, Section 4), and the optimized RCS for $1-\alpha=0.95$ and $p=3$. Also included in each panel is the 
		graph of $f_{\gamma}$ the pdf of $\| \bX \|$.	
		The top, middle and bottom panels are for $\gamma = 1.5$, $\gamma = 2.3$ 
		and $\gamma = 3$, respectively.}
	\label{Figure2}
\end{figure}

\clearpage

\begin{figure}[h!]
	\centering
	\includegraphics[width=1\textwidth]{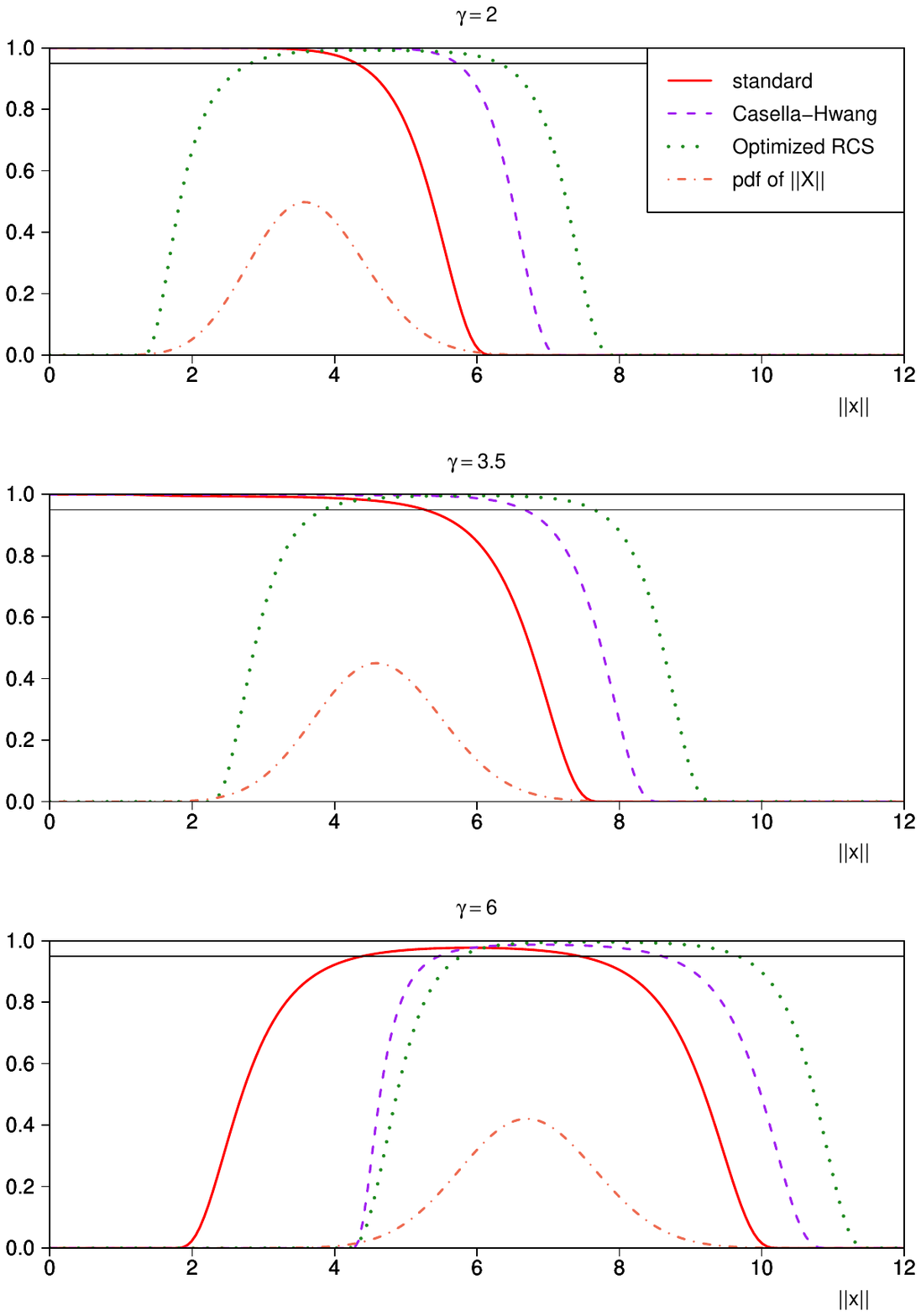} 
	\caption{Each panel consists of graphs of the inclusion functions of the standard confidence sphere, the RCS of Casella and Hwang (1983, Section 4), and the optimized RCS for $1-\alpha=0.95$ and $p=10$. Also included in each panel is the 
		graph of $f_{\gamma}$ the pdf of $\| \bX \|$.	
		The top, middle and bottom panels are for $\gamma = 2$, $\gamma = 3.5$ 
		and $\gamma = 6$, respectively.}
	\label{Figure3}
\end{figure}

\clearpage

\bigskip


\noindent {\large{\bf 4. Results for $\boldsymbol{\sigma^2}$ unknown. Comparison of the
			performances of the optimized RCS and the RCS centered on an analogue of the positive-part
			James-Stein estimator}}

		\medskip
		
		In this section, we consider the more difficult case that $\sigma^2$ is unknown.
		Suppose that we have additional data that provides the estimator $S^2$ for $\sigma^2$,
		where $m S^2/\sigma^2 \sim \chi^2_m$ and $S^2$ and $\bX$ are independent.
		The standard $1-\alpha$ confidence set is
		$\widetilde{I} = \big\{\btheta : \, \| \btheta-\bX \| \leq \widetilde{d} \, S \big \}$,
		where the positive number $\widetilde{d}$ satisfies
		$P\big( G \leq \widetilde{d}^2/p \big) = 1-\alpha$ for $G \sim F_{p,m}$.
		Define the class of RCS's that can be expressed in the form
		\begin{equation*}
		\widetilde{J}(\widetilde{a},\widetilde{b}) = \big\{ \btheta : \| \widetilde{a}(\widetilde{T}) \bX - \btheta \| \leq S \, \widetilde{b}(\widetilde{T}) \big\},
		\end{equation*}
		where $\widetilde{a}:[0, \infty) \rightarrow (0, \infty)$,
		$\widetilde{b}:[0, \infty) \rightarrow (0, \infty)$
		and $\widetilde{T} = \|\bX\|/(\sqrt{p} \, S)$.
		This choice of center and radius has some intuitive appeal, since
		$\widetilde{T} = \|\bX\|/(\sqrt{p} \, S)$ may be viewed as a test statistic for testing the null hypothesis
		that $\btheta = \boldsymbol{0}$ against the
		alternative hypothesis that $\btheta \ne \boldsymbol{0}$.
		This class of RCS's is an analogue, for $\sigma^2$ unknown, of the broad class of
		RCS's described by Casella and Hwang (1983, Section 3), for $\sigma^2$ known.
		We assess the RCS $\widetilde{J}(\widetilde{a},\widetilde{b})$ using both its coverage probability an its
		scaled expected volume, which is defined
		to be the ratio (expected volume of the RCS) / (expected volume of $\widetilde{I}$).

		The coverage probability of $\widetilde{J}(\widetilde{a},\widetilde{b})$ is
		\begin{equation*}
		P \big( \btheta \in \widetilde{J}(\widetilde{a},\widetilde{b}) \big)
		= P \left ( \left \| \widetilde{a} \left( \frac{\|\bY\|}{\sqrt{p} \, W} \right) - \boldsymbol{\vartheta} \right \|
		\le W \widetilde{b} \left( \frac{\|\bY\|}{\sqrt{p} \, W} \right) \right),
		\end{equation*}
		where $W = S/\sigma$, $\boldsymbol{\vartheta} = \btheta/\sigma$ and $\bY = \bX/\sigma$. Obviously,
		$\bY \sim N(\boldsymbol{\vartheta}, \boldsymbol{I})$ and $W$ has the same distribution as $\sqrt{Q/m}$, where $Q \sim \chi^2_m$.
		Since $\bY$ and $W$ are independent, this coverage probability is equal to
		\begin{equation}
		\label{CP1_unknown_sigma_sq}
		\int_0^{\infty} P \left ( \left \| \widetilde{a} \left( \frac{\|\bY\|}{\sqrt{p} \, w} \right) - \boldsymbol{\vartheta} \right \| \le w \, \widetilde{b} \left( \frac{\|\bY\|}{\sqrt{p} \, w} \right) \right) \, f_W(w) \, dw,
		\end{equation}
		where $f_W$ denotes the probability density function of $W$. Let $\gamma = \| \boldsymbol{\vartheta} \| = \| \btheta  / \sigma \|$.
		It follows from Theorem \ref{TheoremA.1} (presented in Appendix
		A) that, for any given $w > 0$ and functions $\widetilde{a}$
		and $\widetilde{b}$,
		\begin{equation}
		\label{part_of_integrand_CP}
		P \left ( \left \| \widetilde{a} \left( \frac{\|\bY\|}{\sqrt{p} \, w} \right) - \boldsymbol{\vartheta} \right \| \le w \, \widetilde{b} \left( \frac{\|\bY\|}{\sqrt{p} \, w} \right) \right)
		\end{equation}
		is a function of $\gamma$. It follows from \eqref{CP1_unknown_sigma_sq} that the coverage probability of $\widetilde{J}(\widetilde{a},\widetilde{b})$ is
		also a function of $\gamma$. We evaluate \eqref{part_of_integrand_CP} using the computationally convenient formula of
		Casella and Hwang (1983, Section 3), which is applicable for $p$ odd. The method used for the numerical evaluation of
		\eqref{CP1_unknown_sigma_sq} is described in Appendix B.

		We define the scaled expected volume of $\widetilde{J}(\widetilde{a},\widetilde{b})$ to be the ratio
		\begin{equation}
		\label{sev_initial_unknown_sigma_squared}
		\frac{E_{\btheta, \,\sigma} (\text{volume of } \widetilde{J}(\widetilde{a},\widetilde{b}))}{E_{\btheta, \,\sigma}(\text{volume of } \widetilde{I})}
		= \frac{E_{\btheta} \left( W^p \, \widetilde{b}^p \left( \frac{\|\bY\|}{\sqrt{p} \, w} \right) \right)}{ \widetilde{d}^p \, E \left( W^p \right)}.
		\end{equation}
		In Appendix 
		B, we show that this is a function of $\gamma = \| \boldsymbol{\vartheta} \| = \| \btheta  / \sigma \|$, for given function $\widetilde{b}$. We also
		derive a new computationally convenient formula for this scaled expected volume.

		Define
		\begin{equation*}
		\widetilde{a}^+(x) =  \max\left\{ 0, \, 1- \left(1-\frac{2}{p}\right)\left(\frac{m}{m+2}\right)\frac{1}{x^2} \right\}.
		\end{equation*}
		Note that $\widetilde{a}^+ \big(\|\bX\|/(\sqrt{p} S) \big) \bX$
		is the positive-part version of an estimator of $\btheta$ due James and Stein (1961, pp. 365--366).
		This estimator belongs to a class of estimators described by Baranchik (1970) and is an analogue, for
		$\sigma^2$ unknown, of the positive-part James-Stein estimator.

		To find the optimized RCS, we require
		that the functions $\widetilde{a}$ and $\widetilde{b}$
		satisfy the following conditions.

		\smallskip
		
		\noindent \underbar{Condition $\widetilde{\text{A}}$} \ $\widetilde{a}:[0, \infty) \rightarrow (0, \infty)$ is a continuous
		nondecreasing function
		that satisfies $\widetilde{a}(x)=\widetilde{a}^{+}(x)$ for all $x \geq k$.
		
		\smallskip
		
		\noindent \underbar{Condition $\widetilde{\text{B}}$}
		\ $\widetilde{b}:[0, \infty) \rightarrow (0, \infty)$ is a continuous
		nondecreasing function
		that satisfies $\widetilde{b}(x) = \widetilde{d}$ for all $x \geq \widetilde{k}$.
		
		\smallskip

		We compare two different optimized RCS's. The first of these RCS's is centered on $\widetilde{a}^+ \big(\|\bX\|/(\sqrt{p} S) \big) \bX$.
		In other words, this RCS has the form  $\widetilde{J}(\widetilde{a}^+,\widetilde{b})$. For computational feasibility, we
		specify the following parametric form for the
		function $\widetilde{b}$.
		\begin{enumerate}
			
			\item[]
			
			Suppose that $y_1,\dots,y_{q_2}$ satisfy $0 = y_1 < y_2 < \dots < y_{q_2} = k$. The function $\widetilde{b}$ is fully
			specified by the vector $\widetilde{b}(y_1),\dots,\widetilde{b}(y_{q_2})$ as follows. The value of $\widetilde{b}(y)$ for any given
			$y \in [0, k]$ is found by piecewise cubic Hermite polynomial interpolation
			for these given function values. We call $y_1, \dots,y_{q_2}$ the knots
			of this piecewise cubic Hermite polynomial.
			
		\end{enumerate}
		\noindent For judiciously-chosen values of $k$ and these knots, we compute the function $\widetilde{b}$, which takes
		this parametric form, is nondecreasing and is
		such that (a) the scaled expected
		volume evaluated at $\btheta = \boldsymbol{0}$ (i.e. at $\gamma= 0$) is minimized
		and (b) the coverage
		probability of $\widetilde{J}(\widetilde{a}^+,\widetilde{b})$ never falls below $1-\alpha$.

		The second of the RCS's has the form $\widetilde{J}(\widetilde{a},\widetilde{b})$. For computational feasibility, we
		additionally specify the following parametric form for the
		function $\widetilde{a}$.
		\begin{enumerate}
			
			\item[]
			
			Suppose that $x_1,\dots,x_{q_1}$ satisfy $0 = x_1 < x_2 < \dots < x_{q_1} = k$. The function $\widetilde{a}$ is fully
			specified by the vector $\widetilde{a}(x_1),\dots,\widetilde{a}(x_{q_1})$ as follows. The value of $\widetilde{a}(x)$ for any given
			$x \in [0, k]$ is found by piecewise cubic Hermite polynomial interpolation
			for these given function values. We call $x_1, \dots,x_{q_1}$ the knots
			of this piecewise cubic Hermite polynomial.

		\end{enumerate}
		\noindent For judiciously-chosen values of $k$ and the knots, we compute the functions $\widetilde{a}$ and $\widetilde{b}$, which take
		these parametric forms, are nondecreasing and are
		such that (a) the scaled expected
		volume evaluated at $\btheta = \boldsymbol{0}$ (i.e. at $\gamma= 0$) is minimized
		and (b) the coverage
		probability of $\widetilde{J}(\widetilde{a},\widetilde{b})$ never falls below $1-\alpha$.

		For $1-\alpha = 0.95$, we compare the coverage probability and scaled expected volume
		of these two RCS's for odd values of $p$.
		We chose the knots of $\widetilde{a}$ that allow this function
		to provide a good approximation to $\widetilde{a}^+$. In this way,
		we sought to ensure that $\widetilde{J}(\widetilde{a},\widetilde{b})$ could perform at least as well as
		$\widetilde{J}(\widetilde{a}^+,\widetilde{b})$
		in terms of minimizing the scaled expected volume at $\gamma = 0$, subject to the
		coverage constraint.
		Some exploratory computations led us to choose $k=10$ and the following knots for
		$\widetilde{a}$ and $\widetilde{b}$.
		Since $\widetilde{a}^+(x)=0$
		for $0 \leq x \leq \sqrt{2(p-2)m/p(m+2)}$, we place the
		first two knots of the function $\widetilde{a}$ at $0$ and $\sqrt{2(p-2)m/p(m+2)}$.
		The next three knots of $\widetilde{a}$ are at equally spaced positions between
		$\sqrt{2(p-2)m/p(m+2)}$ and $k/2$. The last two knots of $\widetilde{a}$ are
		at $k/2$ and $k$. For both RCS's we
		place the knots of the function $\widetilde{b}$ at $0, \,2, \,4, \,6, \,8,$ and $10$.

		The coverage constraint was implemented in the computations by requiring that the
		coverage probability of these RCS's is greater than or equal to $1-\alpha$ for all
		$\gamma \in \{0,1,2,\dots,64,65\}$. This was shown to be adequate to the task by checking
		numerically, at the completion of the computation of these RCS's,
		that the coverage probability constraint
		is satisfied for all $\gamma \ge 0$.

		We compare the two optimized RCS's for all combinations of
		$p \in \{3,5,7,9,25\}$ and $m \in \{3,10,30\}$. Figure \ref{p3m3_sig_unknown}
		compares these RCS's in detail
		for $p=3$ and $m=3$.
		Table \ref{sigUnknown_all_results} presents the comparison of
		the scaled expected volumes  at $\btheta = \boldsymbol{0}$ of the two optimized RCS's
		for all the combinations of $m \in \{3,10,30\}$ and $p \in \{3,5,7,9,25\}$.
		This table shows that, for every combination of $m$ and $p$ considered,
		the RCS of the form $\widetilde{J}(\widetilde{a},\widetilde{b})$ achieves a significantly lower scaled expected volume at $\btheta = \boldsymbol{0}$
		than the RCS of the form $\widetilde{J}(\widetilde{a}^+,\widetilde{b})$.
		
			For these optimized RCS's, the decrease in the scaled expected volume
		at $\btheta = \boldsymbol{0}$ is higher when $m$ is smaller, for given $p$.
		Note that the coverage probability results of these
		RCS's are not presented, since both of these optimized RCS's
		achieve a minimum coverage probability greater than or equal to $0.95$.
		
		In summary, both of the optimized $1-\alpha$ RCS's
		compare favorably with the standard $1-\alpha$ confidence set for $\btheta$.
		Also, the optimized $1-\alpha$ RCS of the form $\widetilde{J}(\widetilde{a},\widetilde{b})$ compares favorably with
		the optimized $1-\alpha$ RCS of the form $\widetilde{J}(\widetilde{a}^+,\widetilde{b})$ in terms of the scaled expected volume
		at $\btheta = \boldsymbol{0}$.

		\newpage
		\begin{figure}[h!]
			\centering
			\begin{tabular}{c}
				\includegraphics[width=0.75\textwidth]{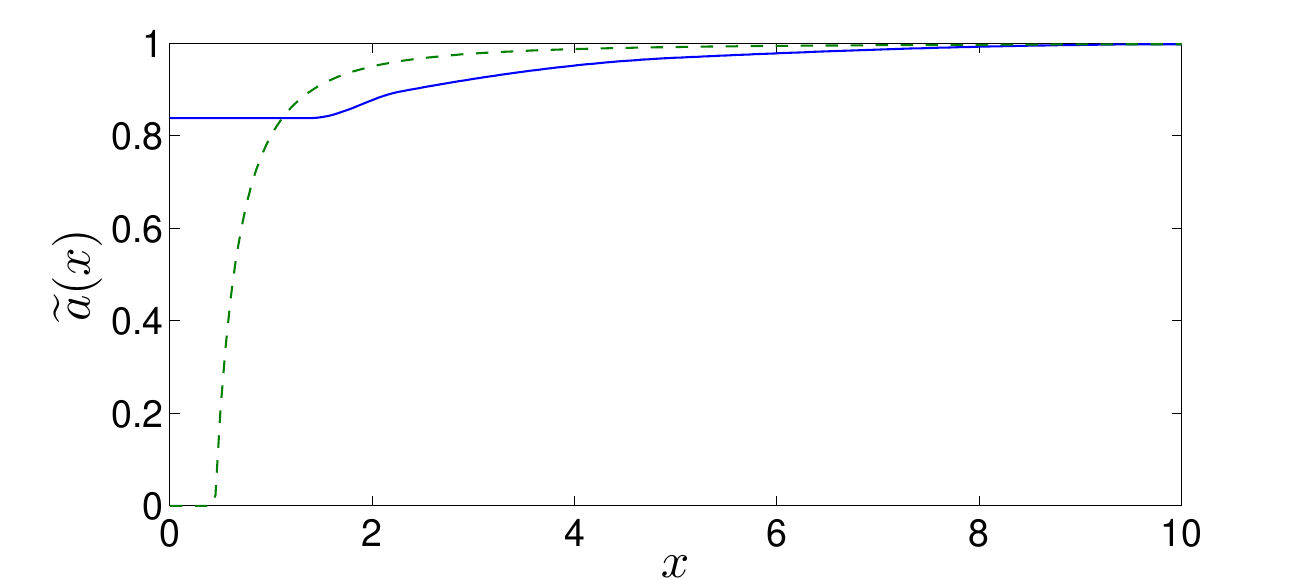} \\
				\includegraphics[width=0.75\textwidth]{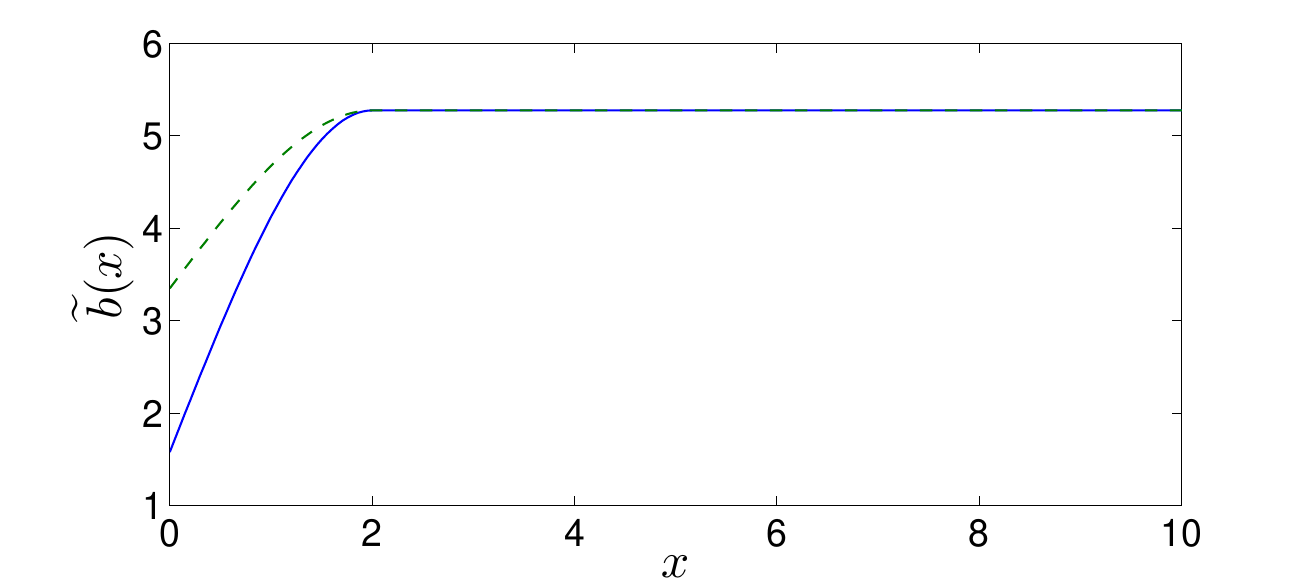}  \\
				\includegraphics[width=0.75\textwidth]{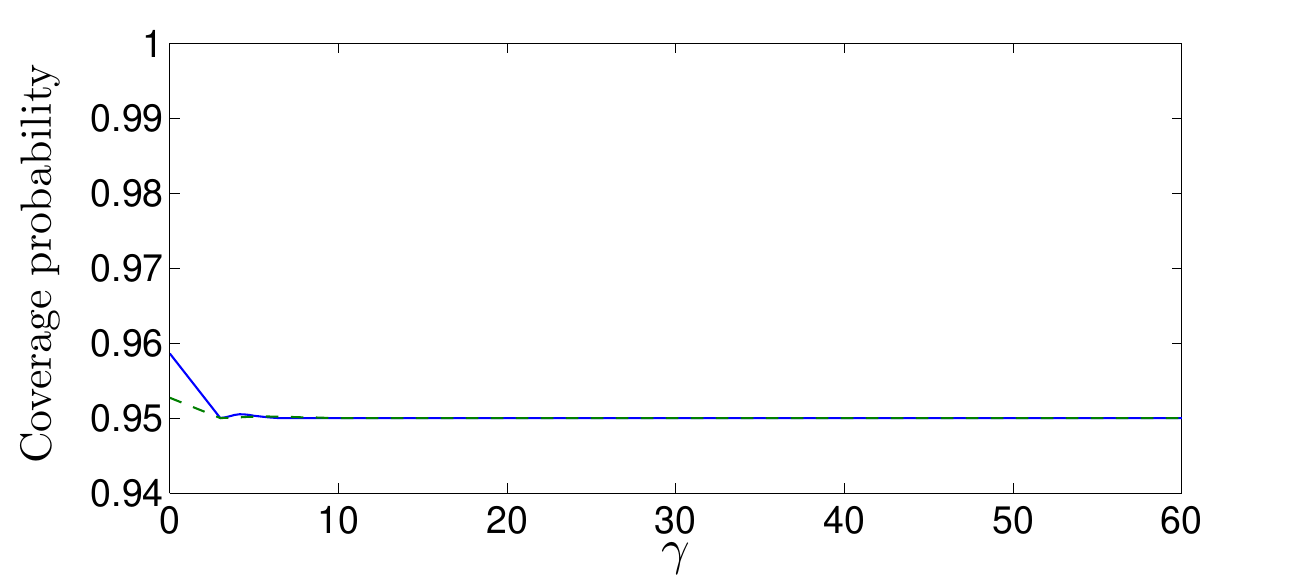} \\
				\includegraphics[width=0.75\textwidth]{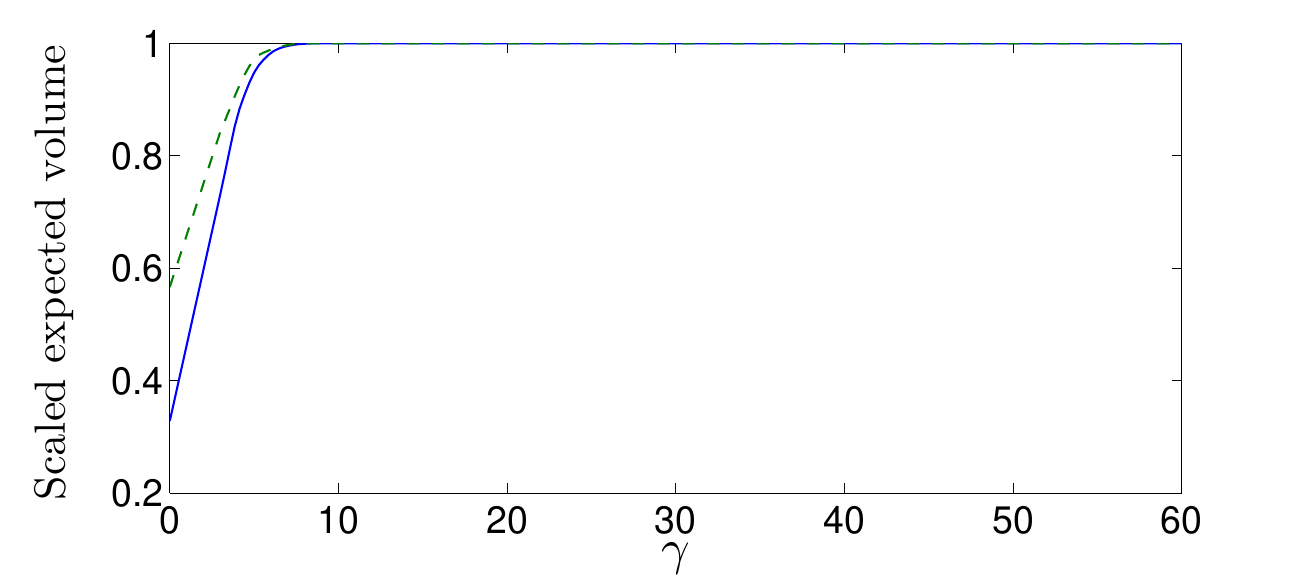} \\
				Legend: \textbf{\textcolor[rgb]{0, 0, 0.8}{------}} optimized RCS of the form $\widetilde{J}(\widetilde{a},\widetilde{b})$
				\ \ \  \textbf{\textcolor[rgb]{0.29, 0.6, 0.13}{- - -}}  optimized RCS of the form $\widetilde{J}(\widetilde{a}^+,\widetilde{b})$
			\end{tabular}
			\caption{Graphs of the functions $\widetilde{a}$ and $\widetilde{b}$ and the coverage probability and the
				scaled expected volume (as functions of $\gamma =  \| \btheta  / \sigma \|$) for both optimized RCS's,
				for $1-\alpha=0.95$, $p=3$ and $m=3$.}
			\label{p3m3_sig_unknown}
		\end{figure}
		\clearpage

		\begin{table}[!htbp] 
			\centering
			\begin{tabular}{>{\centering\arraybackslash}m{1cm} >{\centering\arraybackslash}m{5.5cm} >{\centering\arraybackslash}m{5.5cm}}
				\hline
				& \multicolumn{2}{c}{$m=3$}     \\
				\hline
				\\
				& SEV at $\btheta = \boldsymbol{0}$ of the      & SEV at $\btheta = \boldsymbol{0}$ of the   \\
				$p$ & optimized RCS of the form $\widetilde{J}(\widetilde{a}^+,\widetilde{b})$                & optimized RCS of the form $\widetilde{J}(\widetilde{a},\widetilde{b})$                \\
				\\
				3	& 0.56644                                   & 0.32788               \\
				5	& 0.07008                                   & 0.01724               \\
				7	& 0.00555                                   & 0.00072               \\
				9	& 0.00029                                   & 0.00008               \\
				.	& .                                         & .                     \\
				.	& .                                         & .                     \\
				25	& 3.19 $\times 10^{-7}$                     & 7.93 $\times 10^{-8}$ \\
				\\
				\hline
				& \multicolumn{2}{c}{$m=10$}     \\
				\hline
				\\
				& SEV at $\btheta = \boldsymbol{0}$ of the      & SEV at $\btheta = \boldsymbol{0}$ of the   \\
				$p$ & optimized RCS of the form $\widetilde{J}(\widetilde{a}^+,\widetilde{b})$               & optimized RCS of the form $\widetilde{J}(\widetilde{a},\widetilde{b})$                \\
				\\
				3	& 0.71589                                   & 0.52070               \\
				5	& 0.22627                                   & 0.11229               \\
				7	& 0.05747                                   & 0.02035               \\
				9	& 0.01215                                   & 0.00332               \\
				.	& .                                         & .                     \\
				.	& .                                         & .                     \\
				25	& 4.27 $\times 10^{-6}$                     & 1.12 $\times 10^{-7}$ \\
				\\
				\hline
				& \multicolumn{2}{c}{$m=30$}     \\
				\hline
				\\
				& SEV at $\btheta = \boldsymbol{0}$ of the      & SEV at $\btheta = \boldsymbol{0}$ of the   \\
				$p$ & optimized RCS of the form $\widetilde{J}(\widetilde{a}^+,\widetilde{b})$                & optimized RCS of the form $\widetilde{J}(\widetilde{a},\widetilde{b})$               \\
				\\
				3	& 0.78059                                   & 0.56333               \\
				5	& 0.33636                                   & 0.17951               \\
				7	& 0.12474                                   & 0.05855               \\
				9	& 0.04111                                   & 0.01721               \\
				.	& .                                         & .                     \\
				.	& .                                         & .                     \\
				25	& 3.49 $\times 10^{-6}$                     & 1.13 $\times 10^{-6}$ \\
				\hline
			\end{tabular}
			\caption{Comparison of the optimized RCS's of the forms $\widetilde{J}(\widetilde{a}^+,\widetilde{b})$ and $\widetilde{J}(\widetilde{a},\widetilde{b})$,
				with respect to the scaled expected volume (SEV) at $\btheta = \boldsymbol{0}$,
				for $1-\alpha=0.95$, $p \in \{3,5,7,9,25 \}$ and $m \in \{3,10,30 \}$.}
			\label{sigUnknown_all_results}
		\end{table}

\bigskip

\noindent {\large{\bf 5. Conclusion}}

\medskip

The method of construction of a $1-\alpha$ confidence set that we have used is the following. Suppose that we have a
clearly specified class of confidence sets and a clearly specified criterion that should be optimized.
This specified criterion is numerically
optimized, subject to the coverage constraint and
the constraint that the confidence set (belonging to this class)
has volume no larger than the standard $1-\alpha$ confidence set, for all possible data values.

We have successfully applied this method for the broad class of recentered confidence spheres described by
Casella and Hwang (1983, Section 3), in the case of known $\sigma^2$, and an analogue of this class, in the case of unknown
$\sigma^2$. Motivated by the assumption that we have uncertain prior information that $\btheta = \boldsymbol{0}$,
the criterion that we have chosen to optimize is the scaled expected volume at $\btheta = \boldsymbol{0}$.
This optimization is possible because
a recentered confidence sphere has relatively simple properties. This sphere is specified by two nondecreasing real-valued functions (which, in turn, specify the center and radius functions) defined on the positive
real line. Both the coverage probability and the scaled expected volume of this sphere are readily-computed functions
of the scalar parameter $\gamma = \| \btheta  / \sigma \|$.
This method of construction can also be applied for other criteria. For example, the criterion could be a weighted average (where the weight is a function of $\| \btheta \|$) of the scaled expected volume, with the largest weight at $\btheta = \boldsymbol{0}$.

Confidence sets for the multivariate normal mean with other shapes have been proposed by Faith (1976), Berger (1980),
Shinozaki (1989), Tseng and Brown (1997) and Efron (2006).
Reviews of the literature on confidence sets for the multivariate normal mean
are provided by Efron (2006) and Casella and Hwang (2012).
It would be interesting to know whether or not our method of construction can also
be applied to a confidence set with one of these other shapes.

\bigskip


\noindent {\large{\bf Appendix A: Results for $\boldsymbol{\sigma}^2$ known}}

\medskip

In this appendix, we derive computationally-convenient formulas for the coverage probability and the scaled expected volume
of the RCS $J(a,b)$, when $\sigma^2$ is known. We assume, without loss of generality,
that $\sigma^2 = 1$. Suppose that $p\geq3$. Let $\gamma = \|\btheta\|$.


\medskip

\noindent {\bf A computationally convenient formula
	for the coverage probability of $\boldsymbol{J(a,b)}$}

\medskip

In this section we show that the coverage probability of $J(a,b)$ is an even function of $\gamma$, for given functions $a$ and $b$, and we
derive a computationally convenient formula
for this coverage probability. 
We first present the proofs and derivations and then state the results.

The coverage probability of $J(a,b)$ is
\begin{equation*}
P\big(\btheta \in J(a,b) \big)=P\big( \| \aT \bX - \btheta \| \leq \bT  \big).
\end{equation*}
Let
$\bZ = \bX-\btheta$, so that $\bZ \sim N\left( \boldsymbol{0}, \boldsymbol{I} \right)$.
We write $\bZ = R \bU$ where $R$ and $\bU$ are independent,
$R^2 \sim \chi^2_p$ and $\bU$ is a random $p$-vector which is distributed uniformly on
the surface of a unit sphere in $\mathbb{R}^p$. Then,
$\|\bZ\|^2 = \bZ^{\top}\bZ = R^2\bU^{\top}\bU = R^2$.
For $\btheta = \boldsymbol{0}$, $\|\bX\|^2 = R^2$. Also,
for $\btheta \ne \boldsymbol{0}$,
\begin{align}
\|\bX\|^2 &= \|\btheta+\bZ\|^2 \notag \\
&= (\btheta+\bZ)^{\top}(\btheta+\bZ) \notag \\
&= \|\btheta\|^2 + 2\btheta^{\top}\bZ + \|\bZ\|^2 \notag \\
\label{X-sig}
&= \|\btheta\|^2 + 2\|\btheta\| \|\bZ\| \left(\frac{\btheta}{\|\btheta\|}\right)^{\top}\left(\frac{\bZ}{\|\bZ\|}\right) + \|\bZ\|^2.
\end{align}
Let $L=\left(\btheta/\|\btheta\|\right)^{\top}\left(\bZ/\|\bZ\|\right)$.
Note that $L$ is a random variable
which has a distribution that does not depend on $\btheta$.
Let $f_L$ denote the probability density function of $L$.
Let $B(x,y) = \Gamma(x) \Gamma(y) / \Gamma(x+y)$ denote the beta function.
For $p \geq 3$,
\begin{equation*}
f_L(\ell) =
\begin{cases}
\left(\sqrt{1-\ell^2}\right)^{p-3} \Big/ {B\big(1/2, \, (p-1)/2 \big)} & \text{for} \ \ \  -1 \leq \ell \leq 1 \\
0 & \text{otherwise}.
\end{cases}
\end{equation*}
Now, \eqref{X-sig} can be written as follows.
\begin{equation}
\label{norm_X_squared}
\|\bX\|^2 = \gamma^2 + 2\gamma RL + R^2.
\end{equation}
Note that this formula is valid for all $\gamma \ge 0$ if, for example, we set $L=1$ for $\btheta = \boldsymbol{0}$.
Thus
\begin{equation}
\label{T}
T = \|\bX\|/\sqrt{p} = \sqrt{(\gamma^2 + 2\gamma RL + R^2)/p}.
\end{equation}
For $\btheta \ne \boldsymbol{0}$,
\begin{align*}
P\big(\btheta \in J(a,b) \big)
&= P\big( \| \aT \bX - \btheta \|^2 \leq b^2(T)  \big) \\
&= P \big( a^2(T) \|\bX\|^2 - 2 a(T) \big(\gamma^2 + \gamma R L \big) + \gamma^2 \le b^2(T)  \big).
\end{align*}
This is a function of $\gamma$, by
\eqref{norm_X_squared} and \eqref{T} and the fact that $(R,L)$ has a distribution that does not depend on $\btheta$.
We now derive the new computationally convenient formula for the coverage probability of $J(a,b)$.
By the law of total probability, this coverage probability is equal to
\begin{align*}
& P\left( \| \aT \bX - \btheta \| \leq \bT, T < k  \right)
+ P\left( \| \aT \bX - \btheta \| \leq \bT, T \geq k  \right) \\
& = P\left( \| \aT \bX - \btheta \| \leq \bT, T < k  \right)
+ P\left( \| \aplsT \bX - \btheta \| \leq d, T \geq k  \right) \\
& = P\left( \| \aT \bX - \btheta \| \leq \bT, T < k  \right)
+ P\left( \| \aplsT \bX - \btheta \| \leq d \right) \\
& \ \ \ \ - P\left( \| \aplsT \bX - \btheta \| \leq d, T < k  \right).
\end{align*}
By using the law of total probability in this way, we simplify the computer programming required
for the evaluation of the coverage probability.
Let
\begin{align*}
c(\gamma; a,b) &= P\left( \| \aT \bX - \btheta \| \leq \bT, T < k  \right) \\
c^*(\gamma; a^+) &= P\left( \| \aplsT \bX - \btheta \| \leq d \right) \\
c^+(\gamma; a^+) &= P\left( \| \aplsT \bX - \btheta \| \leq d, T < k  \right).
\end{align*}
Thus, the coverage probability of $J(a,b)$ is equal to $c(\gamma; a,b)+c^*(\gamma; a^+)-c^+(\gamma; a^+)$.
We now derive
computationally convenient approximations for $c(\gamma; a,b)$, $c^*(\gamma; a^+)$ and $c^+(\gamma; a^+)$.

\medskip

\noindent \underline{Derivation of the computationally convenient approximation for $c(\gamma; a,b)$}

\medskip

\noindent Observe that
\begin{align}
\label{aX_theta_onSig}
\left\| \aT \bX - \btheta \right\|^2
&= \left\| \aT \left(\bX-\btheta\right) + (\aT-1) \btheta \right\|^2   \notag\\
&= \left\| \aT \bZ + (\aT-1) \btheta \right\|^2   \notag\\
&= \big[ \aT \bZ + \{\aT-1\} \btheta \big]^\top \big[ \aT \bZ + \{\aT-1\} \btheta \big]   \notag\\
&= a^2(T)\bZ^\top \bZ + \aT\{\aT-1\}\bZ^\top \btheta   \notag\\
& \ \ \ + \{\aT-1\}\aT \btheta^\top \bZ + \{\aT-1\}^2 \btheta^\top \btheta    \notag\\
&= a^2(T)R^2 + 2\aT\{\aT-1\}\gamma R L + \{\aT-1\}^2\gamma^2.
\end{align}

\noindent Let $t = \sqrt{(r^2 + 2\gamma r \ell + \gamma^2)/p}$.
We define functions $g$ and $h$ as follows.
\begin{align*}
g(r,\ell,\gamma) &= t^2 - k^2 \\
h(r,\ell,\gamma;a,b) &= \sqrt{a^2(t)r^2 + 2\at\{\at-1\}\gamma r \ell + \{\at-1\}^2\gamma^2} - \bt.
\end{align*}

\noindent Let $\cal I$ be defined as follows.
\begin{align*}
{\cal I}(\cal A) =
\begin{cases}
1 \ \ \ \  \text{if} \ \ {\cal A} \ \text{is true} \\
0 \ \ \ \ \text{if} \ \ {\cal A} \ \text{is false} ,
\end{cases}
\end{align*}
where $\cal A$ is an arbitrary statement.
By definition,
$c(\gamma; a,b)$ is equal to
\begin{align}
&P\left( \| \aT \bX - \btheta \| \leq \bT, T < k  \right) \notag \\
&= P\left( \sqrt{a^2(T)R^2 + 2\aT(\aT-1)\gamma R L + (\aT-1)^2\gamma^2} \leq \bT, T < k  \right) \notag \\
&= \int_{0}^{\infty} \int_{-1}^{1} \! {\cal I} \big\{ \sqrt{a^2(t)r^2 + 2\at(\at-1)\gamma r \ell + (\at-1)^2\gamma^2} \leq \bt \big\} \notag \\
& \ \ \ \ \ \ \ \ \ \ \ \ \ \ \ {\cal I} \big(t^2<k^2\big) \, f_R(r) \ f_L(\ell) \ d\ell \  dr . \notag \\
\label{CP-A-1_noTrunc}
&= \int_{0}^{\infty} \int_{-1}^{1} \! {\cal I} \big\{ h(r,\ell,\gamma;a,b) \leq 0 \big\} \ {\cal I} \big\{g(r,\ell,\gamma)<0\big\}
\, f_R(r) \ f_L(\ell)  \ d\ell \ dr.
\end{align}
To compute this multiple integral, our next step is to truncate the outer integral. We approximate this multiple integral by
\begin{equation}
\int_{\subsup{l_r}}^{\subsup{u_r}} \int_{-1}^{1} \! {\cal I} \big \{ h(r,\ell,\gamma;a,b) \leq 0 \big \}
\ {\cal I} \big \{g(r,\ell,\gamma) < 0 \big\}
\, f_R(r) \ f_L(\ell)  \ d\ell  \ dr
\label{CP-A-1}
\end{equation}
where, for a specified small positive number $\delta$,
\begin{equation}
\label{defn_l_r}
l_r =
\begin{cases}
0 & \text{for} \ \ \ \ p \le 10 \\
\sqrt{F_p^{-1}(\delta/2)} & \text{for} \ \ \ \ p > 10,
\end{cases}
\end{equation}
and
\begin{equation}
\label{defn_u_r}
u_r = \sqrt{F_p^{-1}(1-\delta/2)}
\end{equation}
and $F_p$ denotes the $\chi^2_p$ distribution function.
The reason for not truncating the integral at the lower endpoint for $p \le 10$ is that there
is little to be gained in this case.
The following lemma provides an upper bound on the error of approximation.

\begin{lemma}
	\label{Lemma1}
	Let $e = \eqref{CP-A-1_noTrunc} - \eqref{CP-A-1}$. For $p \le 10$, $0 \le e \le \delta/2$. Also, for $p > 10$, $0 \le e \le \delta$.
\end{lemma}

\begin{proof}
	Let $e = \eqref{CP-A-1_noTrunc} - \eqref{CP-A-1}$. Obviously, $e \ge 0$ and
	\begin{align*}
	e &= \int_{0}^{\subsup{l_r}} \left[ \int_{-1}^{1} \! {\cal I} \big\{ h(r,\ell,\gamma;a,b) \leq 0 \big\}
	\ {\cal I} \big\{g(r,\ell,\gamma)<0\big\}
	\, f_L(\ell) \, d\ell \right] \, f_R(r) \, dr  \\
	& \ \ \ + \int_{\subsup{u_r}}^{\infty} \left[ \int_{-1}^{1} \! {\cal I} \big\{ h(r,\ell,\gamma;a,b) \leq 0 \big\}
	\ {\cal I} \big\{g(r,\ell,\gamma)<0\big\}
	\, f_L(\ell) \, d\ell \right] \, f_R(r) \, dr.
	\end{align*}
	Since ${\cal I}$ can only take the values $0$ and $1$, we have that
	\begin{align*}
	e &\leq
	\int_{0}^{\subsup{l_r}} \left\{ \int_{-1}^{1} \! f_L(\ell) \, d\ell \right\} \, f_R(r) \, dr
	+ \int_{\subsup{u_r}}^{\infty} \left\{ \int_{-1}^{1} \! f_L(\ell) \, d\ell \right\} \, f_R(r) \, dr \\
	&= \int_{0}^{\subsup{l_r}} \! f_R(r) \, \mathrm{d} r   +   \int_{\subsup{u_r}}^{\infty} \! f_R(r) \, dr \\
	&= P\left( R \leq l_r \right)   +   1 - P\left( R < u_r \right) \\
	&= P\left( R^2 \leq l^2_r \right)   +   1 - P\left( R^2 < u^2_r \right) \\
	&= F_p(l^2_r) + 1 - F_p(u^2_r).
	\end{align*}
\end{proof}

\noindent Obviously, \eqref{CP-A-1} is equal to
\begin{equation}
\label{CP-A-2}
\int_{-1}^{1} \left[ \int_{\subsup{l_r}}^{\subsup{u_r}} \! {\cal I} \big\{ h(r,\ell,\gamma;a,b) \leq 0 \big\} \ {\cal I} \big\{g(r,\ell,\gamma) < 0 \big \}
\, f_R(r) \, dr \right] \  f_L(\ell) \ d\ell.
\end{equation}
We now use the fact that $g(r,\ell,\gamma)$ is a particularly simple function
of $r$ and $\ell$ for given $\gamma$, to simplify \eqref{CP-A-2}.
Note that for given values of $\ell$ and $\gamma$, there is a non-empty interval of values of $r$
such that $g(r,\ell,\gamma)<0$ if and only if
$r^2 + 2\gamma \ell r + \gamma^2 - p k^2 = 0$
has distinct real solutions for $r$.
This condition is equivalent to
\begin{align*}
& (2 \gamma \ell)^2 - 4(\gamma^2-pk^2) > 0 \\
& \Leftrightarrow \ell^2 > (\gamma^2-pk^2)/\gamma^2 \\
& \Leftrightarrow
\begin{cases}
\ell \in [-1,1] \
& \text{if} \  \gamma^2-pk^2 \leq 0 \\
\ell \in [-1,-s) \cup (s,1] \
& \text{otherwise},
\end{cases}
\end{align*}
where $s=\sqrt{\gamma^2 - pk^2}/\gamma$.
This implies that \eqref{CP-A-2} is equal to
\begin{align}
\label{CP-A-3}
c_{\text{approx}}(\gamma; a,b) =
\begin{cases}
\displaystyle \int_{-1}^{1} v(\ell,\gamma; a,b)  f_L(\ell) \ d\ell     \
& \text{if} \  \gamma^2-pk^2 \leq 0 \\
\displaystyle \int_{-1}^{\subsup{-s}} v(\ell,\gamma; a,b)  f_L(\ell) \ d\ell
\\ \qquad \qquad 
+ \displaystyle \int_{\subsup{s}}^{1} v(\ell,\gamma; a,b)  f_L(\ell) \ d\ell    \
& \text{otherwise},
\end{cases}
\end{align}
where
\begin{equation*}
v(\ell,\gamma; a,b) = \displaystyle \int_{\subsup{l_r}}^{\subsup{u_r}} \! {\cal I} \big\{ h(r,\ell,\gamma;a,b) \leq 0 \big\}
\ {\cal I} \big\{ g(r,\ell,\gamma) \le 0 \big\} \, \ f_R(r) \, dr.
\end{equation*}
A computationally convenient formula for $v(\ell,\gamma; a,b)$ 
is obtained as follows.
Suppose the set of values (if it exists) of $r \in [l_r,u_r]$ such that
$g(r,\ell,\gamma) \le 0$ and $h(r,\ell,\gamma;a,b) \leq 0$, for given $\ell$, $\gamma$ and functions $a$ and $b$,
is expressed in the form of a union of disjoint intervals as follows.
\begin{equation*}
\bigcup_{\subsup{i=1}}^{\subsup{K(\ell,\gamma; a,b)}} \Big[l_i(\ell,\gamma; a,b), u_i(\ell,\gamma; a,b)\Big],
\end{equation*}
where $K(\ell,\gamma; a,b)$ denotes the number of such disjoint intervals. Thus
\begin{align*}
v(\ell,\gamma; a,b) =
\begin{cases}
\displaystyle \sum_{\subsup{i=1}}^{\subsup{K(\ell,\gamma; a,b)}} \int_{\subsup{l_i(\ell,\gamma; a,b)}}^{\subsup{u_i(\ell,\gamma; a,b)}}
\!  f_R(r) \, dr \
& \text{if} \ \ K(\ell,\gamma; a,b) > 0 \\
0 \
& \text{if} \ \ K(\ell,\gamma; a,b)=0.
\end{cases}
\end{align*}
Since $R^2 \sim \chi^2_p$, this simplifies to the following. If $K(\ell,\gamma; a,b) > 0$ then
\begin{align*}
v(\ell,\gamma; a,b) =
\displaystyle \sum\limits_{\subsup{i=1}}^{\subsup{K(\ell,\gamma; a,b)}}
\big[ F_p \left\{u_i^2(\ell,\gamma; a,b) \right\} - F_p\left\{l_i^2(\ell,\gamma; a,b) \right\} \big]
\end{align*}
and if  $K(\ell,\gamma; a,b)=0$ then
$v(\ell,\gamma; a,b) = 0$.
Here $F_p$ denotes the $\chi^2_p$ cumulative distribution function.

\medskip

\noindent {\sl Remark:}
The computer program used to find the disjoint intervals $\big[l_i(\ell,\gamma; a,b), \newline u_i(\ell,\gamma; a,b)\big]$,
carries out an extensive grid search, followed by the application of the MATLAB
zero-finding function fzero. This programming is designed to account for the possibility
of quite a large number of such intervals. However, careful investigations suggest that
$h(r,\ell,\gamma;a,b)$, considered as a function of $r$, is smooth,
leading typically to a single interval i.e. $K(\ell,\gamma; a,b)=1$.
A similar remark applies to the computation of the disjoint intervals
involved in computations of the computationally convenient approximations
for $c^*(\gamma; a^+)$ and $c^+(\gamma; a^+)$.

\bigskip

\noindent \underline{Derivation of a computationally convenient approximation for $c^*(\gamma; a^+)$}

\medskip

Similarly to the previous derivation of the computationally convenient approximation
for $c(\gamma; a,b)$, we observe that
\begin{equation*}
\left\| \aplsT \bX - \btheta \right\|^2 = \aplsT^2R^2 + 2\aplsT(\aplsT-1)\gamma R L + (\aplsT-1)^2\gamma^2,
\end{equation*}
and we define the function $h^+$ as follows.
\begin{equation*}
h^{+}(r,\ell,\gamma;a^+) = \sqrt{(\aplst)^2r^2 + 2\aplst(\aplst-1)\gamma r \ell + (\aplst-1)^2\gamma^2} - d,
\end{equation*}
where, as in the previous section, $t = \sqrt{(r^2 + 2\gamma r \ell + \gamma^2)/p}$.
By definition, $c^*(\gamma; a^+)$ is equal to
\begin{align}
\label{CP-B1-noTrunc}
&P\left( \| \aplsT \bX - \btheta \| \leq d \right) \notag \\
&= P\left( \sqrt{\aplsT^2 R^2 + 2\aplsT(\aplsT-1)\gamma R L + (\aplsT-1)^2\gamma^2} \leq d \right) \notag \\
&= \int_{0}^{\infty} \int_{-1}^{1} \! {\cal I} \big\{ \sqrt{\aplsT^2r^2 + 2\aplst(\aplst-1)\gamma r \ell + (\aplst-1)^2\gamma^2} \leq d \big\} \notag \\
& \ \ \ \ \ \ \ \ \ \ \ \ \ \ \ \ f_R(r) \, f_L(\ell) \, dr  \, d\ell. \notag \\
&= \int_{0}^{\infty} \int_{-1}^{1} \! {\cal I} \big\{ h^+(r,\ell,\gamma;a^+) \leq 0 \big\}
\, f_R(r) \, f_L(\ell)  \, d\ell \, dr
\end{align}
To compute this multiple integral, our next step is to truncate the outer integral. We approximate this multiple integral by
\begin{equation}
\label{CP-B1-1}
\int_{\subsup{l_r}}^{\subsup{u_r}}  \int_{-1}^{1} \! {\cal I} \big\{ h^+(r,\ell,\gamma;a^+) \leq 0 \big\}
\, f_R(r) \, f_L(\ell)  \, d\ell \, dr
\end{equation}
where $l_r$ and $u_r$ are given by \eqref{defn_l_r} and \eqref{defn_u_r}, respectively.
The following lemma provides an upper bound on the error of approximation.

\begin{lemma}
	\label{Lemma2}
	Let $e = \eqref{CP-B1-noTrunc} - \eqref{CP-B1-1}$. For $p \le 10$, $0 \le e \le \delta/2$. Also, for $p > 10$, $0 \le e \le \delta$.
\end{lemma}

\noindent The proof of this lemma is omitted, because it is very similar to the proof of Lemma 
\ref{Lemma1}.
Obviously, \eqref{CP-B1-1} is equal to
\begin{equation*}
\int_{-1}^{1} \left[ \int_{\subsup{l_r}}^{\subsup{u_r}} \! {\cal I} \big\{ h^+(r,\ell,\gamma;a^+) \leq 0 \big\}
\, f_R(r) \, dr \right]\, f_L(\ell) \, d\ell.
\end{equation*}
This is equal to
\begin{equation*}
c_{\text{approx}}^*(\gamma; a^+)
= \int_{-1}^{1} v^*(\ell,\gamma; a^+) \, f_L(\ell) \, d\ell,
\end{equation*}
where
\begin{equation*}
v^*(\ell,\gamma; a^+) = \displaystyle \int_{\subsup{l_r}}^{\subsup{u_r}} \! {\cal I} \big\{ h^+(r,\ell,\gamma;a^+) \leq 0 \big\}
\, f_R(r) \, dr.
\end{equation*}
Similarly to the previous derivation, a computationally convenient formula for $v^*(\ell,\gamma; a^+)$
is obtained as follows.
Suppose the set of values (if it exists) of $r \in [l_r,u_r]$ such that
$h^+(r,\ell,\gamma;a^+) \leq 0$, for given $\ell$, $\gamma$ and functions $a$ and $b$
is expressed in the form of a union of disjoint intervals as follows.
\begin{equation*}
\bigcup_{\subsup{i=1}}^{\subsup{K^*(\ell,\gamma; a^+)}} \Big[l^*_i(\ell,\gamma; a^+), u^*_i(\ell,\gamma; a^+)\Big],
\end{equation*}
where $K^*(\ell,\gamma; a^+)$ denotes the number of such disjoint intervals.
Thus
\begin{align*}
v^*(\ell,\gamma; a^+) =
\begin{cases}
\displaystyle \sum_{
	\subsup{i=1}}^{\subsup{K^*(\ell,\gamma; a^+)}} \int_{\subsup{l^*_i(\ell,\gamma; a^+)}}^{\subsup{u^*_i(\ell,\gamma; a^+)}}
\!  f_R(r) \, dr  \
& \text{if} \ \ K^*(\ell,\gamma; a^+) > 0 \\
0 \
& \text{if} \ \ K^*(\ell,\gamma; a^+)=0.
\end{cases}
\end{align*}
Since $R^2 \sim \chi^2_p$, this simplifies to the following. If $K^*(\ell,\gamma; a^+) > 0$ then
\begin{align*}
v^*(\ell,\gamma; a^+) =
\displaystyle \sum\limits_{\subsup{i=1}}^{\subsup{K^*(\ell,\gamma; a^+)}}
\Big[ F_p\left\{(u^{*}_i(\ell,\gamma; a^+))^2 \right\} - F_p\left\{(l^{*}_i(\ell,\gamma; a^+))^2 \right\} \Big]
\end{align*}
and if $K^*(\ell,\gamma; a^+)=0$ then 
$v^*(\ell,\gamma; a^+) = 0$.
Here, as in the previous section, $F_p$ denotes the $\chi^2_p$ distribution function.

\medskip

\noindent \underline{Derivation of a computationally convenient approximation for $c^+(\gamma; a^+)$}

\medskip

Note that $c^+(\gamma; a^+)$ is obtained by
simply replacing the function $a$ by
the function $a^+$ and the function $b$ by the constant $d$ in the expression for $c(\gamma; a,b)$.
Thus, by replacing the function $a$ by the function $a^+$ and the function $b$ by the constant $d$
in the expression for the computationally convenient approximation $c_{\text{approx}}(\gamma; a,b)$
for $c(\gamma; a,b)$,
we obtain the computationally convenient approximation $c_{\text{approx}}^+(\gamma; a^+)$ for
$c^+(\gamma; a^+)$.

\bigskip
The following theorem provides a computationally convenient expression for the coverage probability of $J(a,b)$
for $p\geq3$.
\begin{theorem}
	\label{TheoremA.1}
	Suppose that $p\geq3$.
	The coverage probability of $J(a,b)$ is equal to
	\begin{equation}
	\label{CP}
	c(\gamma; a,b) + c^*(\gamma; a^+) - c^+(\gamma; a^+).
	\end{equation}
	An approximation to this coverage probability is
	\begin{equation}
	\label{CP_approx}
	c_{\text{approx}}(\gamma; a,b) + c_{\text{approx}}^*(\gamma; a^+) - c_{\text{approx}}^+(\gamma; a^+),
	\end{equation}
	where the accuracy of this approximation is determined, through
	\eqref{defn_l_r}
	and
	\eqref{defn_u_r}, by the specified small positive number $\delta$.
	The error of approximation $\eqref{CP} - \eqref{CP_approx}$
	lies (a) between $-\delta/2$ and $\delta$, for $p \le 10$ and (b) between $-\delta$ and $2 \delta$, for
	$p > 10$.
\end{theorem}

\medskip

\noindent \textbf{Comparison of the two computationally convenient formulas for the coverage probability
	of $\boldsymbol{J(a,b)}$}

\medskip

Note that there
is a typographical error in this formula as stated on page 691 of Casella and Hwang (1983).
The $(n+1)!$ on the second line of (3.10) should be replaced by $(n+i)!$.
The main advantage of computationally convenient formula stated in Theorem  
\ref{TheoremA.2}
is that it is applicable
for any $p \ge 3$, whereas the computationally convenient formula stated by Casella and Hwang (1983, Section 3), is
applicable only for odd values of $p \ge 3$. We found that the coverage probability computed using the formula
of Casella and Hwang (1983, Section 3), was inaccurate when $\gamma = \|\btheta\|$ is very close to zero but not equal to zero.
There is no such problem with the formula for the coverage probability stated in Theorem \ref{TheoremA.1}.
In terms of computational speed, for small values of $p$, both of these computationally convenient formulas perform
equally well. However, for large values of $p$, the coverage probability can be computed faster
using the formula of Casella and Hwang (1983, Section 3).

\medskip

\noindent {\bf Numerical evaluation of the coverage probability of $\boldsymbol{J(a,b)}$
	using the new computationally convenient formula}

\medskip

The computer programs for the computation of the coverage probability using \eqref{CP_approx} were checked
for correctness in two ways, for some particular examples. Firstly,  the coverage probabilities
computed using the computationally convenient formula of Casella and Hwang (1983) for $p$ odd
and the computationally convenient formula \eqref{CP_approx} were compared.
Secondly, coverage probabilities computed using these formulas were compared with the results
of coverage probabilities computed using
Monte Carlo simulations.

\medskip

\noindent {\bf A computationally convenient formula
	for the scaled expected volume of $\boldsymbol{J(a,b)}$}

\medskip

The following theorem provides a computationally convenient-formula for the scaled expected volume
of the recentered confidence sphere $J(a,b)$.

\begin{theorem}
	\label{TheoremA.2}	
	For given function $b$, the scaled expected volume of $J(a,b)$ is a function of $\gamma = \ntheta$.
	
	\begin{enumerate}
		
		\item
		
		Let $f \big(v;p,\gamma^2 \big)$ denote the noncentral $\chi^2$ probability density function with $p$ degrees of freedom and noncentrality
		parameter $\gamma^2$, evaluated at $v$.
		The scaled expected volume of $J(a,b)$ is
		\begin{equation}
		\label{first_sev}
		\int_0^{\infty} \left\{ \frac{b \big(\sqrt{v/p} \big)}{d} \right\}^p f \big(v;p,\gamma^2 \big) \, dv.
		\end{equation}

		\item
		
		Suppose that $b(x)=d$ for all $x \ge k$, where $k$ is a specified positive number.
		The scaled expected volume of $J(a,b)$ is
		\begin{equation}
		\label{second_sev}
		\int_{\subsup{0}}^{\subsup{pk^2}} \left\{ \frac{b \big(\sqrt{v/p} \big)}{d} \right\}^p f \big(v;p,\gamma^2 \big) \, dv
		+ 1 - F \big(pk^2; p, \gamma^2 \big),
		\end{equation}
		where $F \big(v; p, \gamma^2 \big)$ denotes the noncentral $\chi^2$ cumulative distribution function
		with $p$ degrees of freedom and noncentrality
		parameter $\gamma^2$, evaluated at $v$.

	\end{enumerate}
	
\end{theorem}

\begin{proof}
	
	Note that $V = \|\bX\|^2 = X_1^2 + \dots + X_p^2$ has a noncentral $\chi^2$
	distribution with $p$ degrees of freedom and noncentrality
	parameter $\gamma^2$. It follows from \eqref{sev_initial}
	that the scaled expected volume of $J(a,b)$ is  \eqref{first_sev}.
	Clearly, \eqref{first_sev} is an even function of $\gamma$,
	for given function $b$.
	Suppose that $b(x)=d$ for all $x \ge k$, where $k$ is a specified positive number.
	Obviously,  \eqref{second_sev} follows immediately from \eqref{first_sev}.
	
\end{proof}

\medskip

\noindent {\bf Numerical evaluation of the scaled expected volume using the computationally convenient formula  \eqref{second_sev}}

\medskip

As stated in Section 
2, we suppose that the
function $b$ is a piecewise cubic Hermite interpolating polynomial
in the interval $[0,k]$,
with knots at $y_1, \dots, y_{q_2}$
($0 = y_1 < y_2 < \dots < y_{q_2} = k$) and that $b(x)=d$ for all $x \ge k$.
This function is very smooth between successive knots (it is a cubic
between these knots). However, it may not possess a second derivative at each of the knots. For this
reason, we numerically evaluate \eqref{second_sev} using
the formula
\begin{equation*}
\sum_{\subsup{i=1}}^{\subsup{q_2-1}} \int_{\subsup{p y_i^2}}^{\subsup{p y_{i+1}^2}}
\left\{ \frac{b \big(\sqrt{y/p} \big)}{d} \right\}^p f \big(y;p,\gamma^2 \big) \,  dy + 1 - F \big(pk^2; p, \gamma^2 \big),
\end{equation*}
where each integral is computed separately by numerical quadrature.
The computer programs for the computation of the scaled expected length using \eqref{second_sev} were checked
for correctness, for some particular examples, by comparison with the scaled expected length computed using Monte Carlo
simulation.

\bigskip

\noindent {\large{\bf Appendix B: Results for $\boldsymbol{\sigma}^2$ unknown}}

\medskip

In this appendix, we describe the method used for the numerical evaluation of
\eqref{CP1_unknown_sigma_sq},
the coverage probability of $\widetilde{J}(\widetilde{a},\widetilde{b})$. We also
derive a computationally convenient formula for the scaled expected volume
of $\widetilde{J}(\widetilde{a},\widetilde{b})$.  Suppose that $p\geq3$.
Let $\gamma = \| \boldsymbol{\vartheta} \| = \| \btheta  / \sigma \|$.

\medskip


\noindent {\bf Numerical evaluation of the coverage probability of $\boldsymbol{\widetilde{J}(\widetilde{a},\widetilde{b})}$
	using \eqref{CP1_unknown_sigma_sq}}

\medskip

The optimized RCS is found by numerically solving the constrained optimization problem described in Section 3.
This type of computation has been carried out in other related problems by
Farchione and Kabaila (2008, 2012) and Kabaila and Giri (2009, 2013). The main
lesson from these related computations is that the coverage probability needs to
be computed with great accuracy.

Let
\begin{equation*}
\psi \big(w, \gamma; \widetilde{a}, \widetilde{b} \big)
= P \left \{ \left \| \widetilde{a} \left( \frac{\|\bY\|}{\sqrt{p} \, w} \right) - \boldsymbol{\vartheta} \right \| \le w \, \widetilde{b} \left( \frac{\|\bY\|}{\sqrt{p} \, w} \right) \right\}.
\end{equation*}
Therefore, the formula \eqref{CP1_unknown_sigma_sq} for the coverage probability of $\widetilde{J}(\widetilde{a},\widetilde{b})$ is
\begin{equation*}
\int_0^{\infty} \psi \big(w, \gamma; \widetilde{a}, \widetilde{b} \big) \, f_W(w) \, dw.
\end{equation*}
We numerically evaluate this integral as follows. This integral is equal to
\begin{equation*}
\int_0^1 \psi \big(w, \gamma; \widetilde{a}, \widetilde{b} \big) \, f_W(w) \, dw
+ \int_1^{\infty} \psi \big(w, \gamma; \widetilde{a}, \widetilde{b} \big) \, f_W(w) \, dw.
\end{equation*}
We transform the variable of integration in the second integral from $w$ to $x = F_W(w)$, where $F_W$ denotes the
cumulative distribution function of $W$. Therefore, the coverage probability of $\widetilde{J}(\widetilde{a},\widetilde{b})$ is
\begin{equation*}
\int_0^1 \psi \big(w, \gamma; \widetilde{a}, \widetilde{b} \big) \, f_W(w) \, dw
+ \int_{F_W(1)}^1 \psi \big\{F_W^{-1}(x), \gamma; \widetilde{a}, \widetilde{b} \big\} \, dx,
\end{equation*}
where $\psi \big\{F_W^{-1}(x), \gamma; \widetilde{a}, \widetilde{b} \big\}$ evaluated at $x=1$ is defined to be the limit
as $x \uparrow 1$ of this function. This limit is 1. The integrands in both of these integrals are smooth.
These integrals were computed using Simpson's rule with an appropriately chosen fixed number of evaluations of the integrand.
To help ensure accurate computation of the integrands for both integrals, progressive
numerical quadrature, using Simpson's rule was used and a doubling of equal-length segments
at each stage of the progression is used.
Progressive numerical quadrature is described, for example, in Section 6.1 of Davis and Rabinowitz (1984).
The main stopping criterion is that $|Q_{2s} - Q_s|/Q_{2s} \le 10^{-8}$,
where $Q_s$ denotes the computed quadrature using $s$ segments.
The computer programs for the computation of the coverage probability were checked
for correctness by comparing computed values, for some particular examples, with the
results of coverage probabilities computed using
Monte Carlo simulations.

\medskip

\noindent {\bf A computationally convenient formula
	for the scaled expected volume of $\boldsymbol{\widetilde{J}(\widetilde{a},\widetilde{b})}$}

\medskip

The following theorem provides a computationally convenient formula for the scaled expected volume
of the RCS $\widetilde{J}(\widetilde{a},\widetilde{b})$.

\begin{theorem}
	
	For given function $\widetilde{b}$, the scaled expected volume of $\widetilde{J}(\widetilde{a},\widetilde{b})$
	is a function of $\gamma = \| \boldsymbol{\vartheta} \| = \| \btheta  / \sigma \|$.
	Let $f \big(v;p,\gamma^2 \big)$ and $F \big(v; p, \gamma^2 \big)$ denote the probability density function and the cumulative distribution function, respectively,
	of the noncentral $\chi^2$ distribution with $p$ degrees of freedom and noncentrality parameter $\gamma^2$, evaluated at $v$.
	The scaled expected volume
	of $\widetilde{J}(\widetilde{a},\widetilde{b})$ is
	equal to 1 plus
	\begin{align}
	\label{SEVSigmaUnknown}
	1 +
	\mu^{-1} \int_{0}^{\infty}  \Bigg (\int_{0}^{\pkstyle{p\, k^2 w^2}} \,\left [\frac{\widetilde{b}\left\{ v \big/ (\sqrt{p}\,w ) \right\}}{d} \right]^{p}
	&f(v; p,\gamma^2) \,  dv - \, F(p\, k^2 w^2; p,\gamma^2) \Bigg) 
	\notag
	\\
	&\times w^{p} \, f_W(w) \, dw,
	\end{align}
	%
	%
	%
	where $\mu=\displaystyle{(2/m)^{p/2} \,\, \Gamma\big\{ (p+m)/2 \big\} \big/ \Gamma\left( m/2 \right)}$.
	
\end{theorem}

\begin{proof}
	
	Our proof proceeds from the expression \eqref{sev_initial_unknown_sigma_squared}
	for the scaled expected volume of $\widetilde{J}(\widetilde{a},\widetilde{b})$.
	Since $W$ has the same distribution as $\sqrt{Q/m}$ where $Q \sim \chi^2_m$,
	\begin{align*}
	E \left( W^{p} \right) & = E \left\{ (Q/m)^{p/2} \right\} \\
	& = \frac{1}{m^{p/2}} \,\, E \left( Q^{p/2} \right)   \\
	& = (2/m)^{p/2} \,\, \Gamma\left( \frac{p+m}{2} \right) \Big/ \Gamma\left( \frac{m}{2}\right).
	\end{align*}
	Let $\mu=\displaystyle{(2/m)^{p/2} \,\, \Gamma\left( (p+m)/2 \right) \big/ \Gamma\left( m/2 \right)}$.
	Thus \eqref{sev_initial_unknown_sigma_squared} is equal to
	\begin{align}
	\label{sevSigUnknown_2-1}
	\mu^{-1} \times E_{\btheta} \left( W^{p} \,\, \left[\frac{\widetilde{b}\left\{ ||\bY|| \big/ (\sqrt{p}\,W )\right\}}{d} \right]^{p}  \right),
	\end{align}
	where $\bY = \bX/\sigma$, so that
	$\bY \sim N(\boldsymbol{\vartheta}, \boldsymbol{I})$ for $\boldsymbol{\vartheta} = \btheta/\sigma$.
	Let $V=||\bY||^2$. Note that $V$ has a noncentral $\chi^2$ distribution
	with $p$ degrees of freedom and noncentrality parameter $\gamma^2$.
	Thus, \eqref{sevSigUnknown_2-1} is equal to
	\begin{align}
	\label{sevSigUnknown_3}
	\mu^{-1} \int_{0}^{\infty}  \int_{0}^{\infty} w^{p} \left[\frac{\widetilde{b}\left\{ v \big/ (\sqrt{p}\,w) \right\}}{d} \right]^{p}
	f(v; p,\gamma^2) f_W(w) \, dv \, dw.
	\end{align}

	By Condition $\widetilde{\text{B}}$,
	$\widetilde{b}(x)=d$ for all $x \geq k$. Note that
	$v \big/ \sqrt{p}\,w \geq k$ is equivalent to $v \geq p\, k^2 w^2 $.
	Therefore, \eqref{sevSigUnknown_3} is equal to
	\begin{align*}
	& \mu^{-1} \int_{0}^{\infty}  \int_{0}^{\pkstyle{p\, k^2 w^2}} w^{p} \left[\frac{\widetilde{b}\left\{ v \big/ (\sqrt{p}\,w )\right\}}{d} \right]^{p}
	f(v; p,\gamma^2) f_W(w) \, dv \, dw   \notag \\
	%
	& + \mu^{-1} \int_{0}^{\infty}  \int_{\pkstyle{p\, k^2 w^2}}^{\infty} w^{p}
	f(v; p,\gamma^2) \, f_W(w) \, dv \, dw.
	\end{align*}
	The second term in this expression is equal to
	\begin{align*}
	&\mu^{-1} \int_{0}^{\infty}  \left\{ \int_{\pkstyle{p\, k^2 w^2}}^{\infty}  f(v; p,\gamma^2) \, dv \right\} \, w^{p} \, f_W(w) \, dw \notag\\
	& = \mu^{-1} \int_{0}^{\infty}  \left\{ 1 - F(p\, k^2 w^2; p,\gamma^2) \right\} w^{p} f_W(w) \, dw \notag\\
	& = \mu^{-1} \int_{0}^{\infty} w^{p} f_W(w) \, dw - \mu^{-1} \int_{0}^{\infty}  F(p\, k^2 w^2; p,\gamma^2) \, w^{p} \, f_W(w) \, dw \notag\\
	& = 1 - \mu^{-1} \int_{0}^{\infty}  F(p\, k^2 w^2; p,\gamma^2) \, w^{p} \, f_W(w) \, \mathrm{d} w,
	\end{align*}
	Thus the scaled expected volume
	of $\widetilde{J}(\widetilde{a},\widetilde{b})$ is is equal to \eqref{SEVSigmaUnknown}.
	
\end{proof}

\medskip

\noindent {\bf Numerical evaluation of the scaled expected volume using the computationally convenient formula  \eqref{SEVSigmaUnknown}}

\medskip

To compute the scaled expected volume using \eqref{SEVSigmaUnknown}, we truncate the outer integral.
As before, let $\mu=\displaystyle{(2/m)^{p/2} \,\, \Gamma\left( (p+m)/2 \right) \big/ \Gamma\left( m/2 \right)}$.
We approximate
\eqref{SEVSigmaUnknown} by
\begin{align}
\label{SEVSigmaUnknown_truncated}
1 +
\mu^{-1} \int_{\pkstyle{l_w}}^{\pkstyle{u_w}}  \Bigg (\int_{0}^{\pkstyle{p\, k^2 w^2}} \,\left [\frac{\widetilde{b}\left\{ v \big/ (\sqrt{p}\,w ) \right\}}{d} \right]^{p}
&f(v; p,\gamma^2) \,  dv - \, F(p\, k^2 w^2; p,\gamma^2) \Bigg) 
\notag
\\
&\times w^{p} \, f_W(w) \, dw,
\end{align}
where, for a specified small positive number $\delta$,
\begin{equation*}
l_w =
\begin{cases}
0 & \text{for} \ \ \ \ m+p \le 10 \\
\sqrt{F_{m+p}^{-1} \{\delta/(2 \mu)\}} & \text{for} \ \ \ \ m+p > 10
\end{cases}
\end{equation*}
and
\begin{equation*}
u_w = \sqrt{F_{m+p}^{-1}\{1-\delta/(2 \mu) \}}
\end{equation*}
and $F_{m+p}$ denotes the $\chi^2_{m+p}$ cumulative distribution function. The following lemma provides an upper bound on the error of approximation.

\begin{lemma}
	Let $e = \eqref{SEVSigmaUnknown} - \eqref{SEVSigmaUnknown_truncated}$. For $m+p \le 10$, $0 \le e \le \delta/2$. Also, for $m+p > 10$, $0 \le e \le \delta$.
\end{lemma}

\noindent The proof of this lemma is omitted for the sake of brevity.

\bigskip


\noindent {\large \textbf{References}}

\smallskip

\rf Abeysekera, W. (2014) New Recentered Confidence Spheres for the Multivariate Normal Mean. Unpublished PhD thesis, Department of Mathematics and Statistics, La Trobe University.

\smallskip

\rf Baranchik, A. (1970) A family of minimax estimators of the mean of a multivariate normal
  distribution. {\sl The Annals of Mathematical Statistics}, \textbf{41}, 642--645.

\smallskip

\rf Berger, J. (1980) A robust generalized Bayes estimator and confidence region
for a multivariate normal mean. {\sl Annals of Statistics}, \textbf{8}, 716--761.





\smallskip

\rf Casella, G. and Hwang, J.T. (1983) Empirical Bayes confidence sets for the mean of a multivariate
normal distribution. {\sl Journal of the American Statistical Association}, \textbf{78}, 688--698.

\smallskip

\rf Casella, G. and Hwang, J.T. (1986) Confidence sets and the Stein effect. {\sl Communications in Statistics --
Theory and Methods}, \textbf{15}, 2043--2063.

\smallskip

\rf Casella, G. and Hwang, J.T. (1987) Employing vague prior information in the construction of
confidence sets. {\sl Journal of Multivariate Analysis}, \textbf{21}, 79--104.

\smallskip

\rf Casella, G., Hwang, J.T. and 
Robert, C. (1993) A paradox in decision-theoretic interval estimation. \textsl{Statistica Sinica} \textbf{3}, 141--155.

\smallskip

\rf Casella, G. and Hwang, J.T. (2012) Shrinkage
confidence procedures. {\sl Statistical Science}, \textbf{27}, 51--60.

\smallskip

\rf Davis, P.J. and Rabinowitz, P. (1984) Methods of Numerical Integration, 2nd ed. Academic Press, San Diego, CA.

\smallskip

\rf Efron, B. (2006) Minimum volume confidence regions for a multivariate normal mean vector.
{\sl Journal of the Royal Statistical Society, Series B}, \textbf{68}, 655--670.

\smallskip

\rf Farchione, D. and Kabaila, P. (2008) Confidence intervals for the normal mean
utilizing prior information. {\sl Statistics \& Probability Letters}, \textbf{78}, 1094--1100.

\smallskip

\rf Farchione, D. and Kabaila, P. (2012) Confidence intervals in regression centred
on the SCAD estimator. {\sl Statistics \& Probability Letters}, \textbf{82}, 1953--1960.

\smallskip

\rf Faith, R.E. (1976)
Minimax Bayes point and set estimators of a multivariate normal mean.
Unpublished PhD thesis, Department of Statistics, University of Michigan.

\smallskip

\rf Fritsch, F.N. and Carlson, R.E. (1980) Monotone piecewise cubic interpolation. {\sl SIAM J. Numerical Analysis} \textbf{17}, 238--246.

\smallskip

\rf Hodges, J.L. and Lehmann, E.L. (1952) The use of previous experience in reaching statistical decisions. {\sl Annals of Mathematical Statistics} \textbf{23}, 396--407.

\smallskip

\rf Hwang, J.T.  and  Casella, G.(1982) Minimax
confidence sets for the mean of a multivariate normal distribution.
{\sl Annals of Statistics}, \textbf{10}, 868--881.

\smallskip

\rf James, W. and Stein, C. (1961) Estimation with quadratic loss. {\sl In
Proceedings of the Fourth Berkeley Symposium on Mathematical
  Statistics and Probability}, Vol. 1, pp. 361--379. Berkeley: University of California Press.



\smallskip

\rf Kabaila, P. (2013) Note on a paradox in decision-theoretic interval estimation.
{\sl Statistics \& Probability Letters}, \textbf{83}, 123--126.

\smallskip

\rf Kabaila, P. and Giri, K. (2009) Confidence intervals in regression utilizing
uncertain prior information. {\sl Journal of Statistical Planning and Inference},
\textbf{139}, 3419--3429.



\smallskip

\rf Kabaila, P. and Giri, K. (2013) Further properties of frequentist confidence intervals
in regression that utilize
uncertain prior information. {\sl Australian \& New Zealand Journal of Statistics},
\textbf{55}, 259--270.

\smallskip

\rf Kabaila, P. and Tissera, D. (2014) Confidence intervals
in regression that utilize uncertain prior information about a vector parameter. \textsl{Australian \& New Zealand Journal of Statistics}
\textbf{56}, 371--383.

\smallskip

\rf Samworth, R. (2005) Small confidence sets for the mean of a spherically
symmetric distribution. {\sl Journal of the Royal Statistical Society, Series B},
\textbf{67}, 343--361.

\smallskip

\rf Shinozaki, N. (1989) Improved confidence sets for the mean of a multivariate
normal distribution. {\sl Annals of the Institute of Mathematical Statistics}, \textbf{41},
331--346.

\smallskip

\rf Stein, C. (1962) Confidence sets for the mean of a multivariate normal
distribution.
{\sl Journal of the Royal Statistical Society, Series B}, \textbf{9}, 1135--1151.

\smallskip

\rf Tseng, Y.L. and Brown, L.D. (1997) Good exact confidence sets for a multivariate normal mean.
{\sl Annals of Statistics}, \textbf{5}, 2228--2258.

\end{document}